\documentclass[11pt]{article}

\newcommand{\choosesecondifkeepwired}[2]{#2}

\newcommand{\omitproofs}[1]{#1}

\usepackage{floatflt}
\usepackage{color}
\usepackage{graphicx}
\usepackage{mdwlist}
\usepackage{fullpage}
\usepackage[ruled, vlined]{algorithm2e} % [section] is use to define the numbering mode

\newcommand{\qed}{\hspace*{\fill}\medskip}

\newenvironment{descriptionsmallermargin} {
\begin{basedescript}{\desclabelwidth{.10in}\setlength{\topsep}{2pt}\setlength{\itemsep}{0pt}  } }
{\end{basedescript}}

\newcounter{foo}
\newenvironment{myenumerate}{\begin{list}{\arabic{foo}.}{
        \usecounter{foo}\setlength{\leftmargin}{0.25in}
        \setlength{\topsep}{1pt}
        \setlength{\itemsep}{0pt}
}}{ \end{list}}

\newenvironment{mylist}{\begin{list}{$\bullet$}
    {   \setlength{\itemsep}{0pt}
        \setlength{\topsep}{0pt}}
    }
{\end{list}}

\usepackage{times}

\usepackage{amsmath}

\usepackage{listings}
\usepackage{url}
\urldef{\mailsa}\path|{lidan, amol}@cs.umd.edu|

\newcounter{examplectr}
\newenvironment{example}{\em {\noindent\bf \em Example\stepcounter{examplectr}
\theexamplectr:}}{}

\newcommand{\topic}[1]{\vspace{4pt} \noindent \underline{\bf #1}}
\newcommand{\commentout}[1]{}
\newcommand{\calX}{{\mathcal X}}

\newcommand{\calG}{{\mathcal G}}

\newcommand{\calM}{{\mathcal M}}
\newcommand{\calT}{\mathcal{T}}
\newcommand{\vcalT}{\vec{\mathcal{T}}}
\newcommand{\calF}{\mathcal{F}}

\newcommand{\TS}{\mathcal{TS}}
\newcommand{\dist}{\mathrm{d}}
\newcommand{\davg}{\mathrm{d}_{avg}}
\newcommand{\cost}{\mathtt{cost}}
\newcommand{\opt}{OPT}
\newcommand{\ropt}{\widetilde{OPT}}

\newenvironment{proof}{\noindent\textbf{Proof: }\ignorespaces}{}

\newcommand{\eat}[1] {}
\newcommand{\ignore}[1] {}

\newtheorem{Theorem}{Theorem}
\newtheorem{Lemma}{Lemma}

\newtheorem{Definition}{Definition}

\title{On Computing Compression Trees for Data Collection in \choosesecondifkeepwired{Wireless}{} Sensor Networks}
\author{
\begin{tabular}{ccccc}
    Jian Li & \mbox{\ \ \ \ } & Amol Deshpande &\mbox{\ \ \ \ } &  Samir Khuller\\[3pt]
\multicolumn{5}{c}{\{lijian, amol, samir\}@cs.umd.edu} \\[3pt]
\multicolumn{5}{c}{University of Maryland at College Park} \\
\end{tabular}
}

\begin{document}

%\institute{Computer Science Department, University of Maryland,\\
%A.V. Williams Building, College Park, MD 20742, U.S.A.\\
%\mailsa\\
%}

\maketitle

\begin{abstract}
%\color{black}
We address the problem of efficiently gathering correlated data from a wired or a wireless sensor network, with
the aim of designing algorithms with provable optimality guarantees, and understanding how close we can get
to the known theoretical lower bounds.
Our proposed approach is based on finding an optimal or a near-optimal {\em compression tree} for a given
sensor network:
a compression tree is a directed tree over the sensor network nodes such that the value of
a node is compressed using the value of its parent.
We consider this problem under different
communication models, including the {\em broadcast communication} model that enables many
new opportunities for energy-efficient data collection.
We draw connections between the data collection problem and
a previously studied graph concept, called {\em weakly connected dominating sets}, and we use this
to develop novel approximation algorithms for the problem. We
present comparative results on several synthetic and real-world datasets showing that our algorithms
construct near-optimal compression trees that yield a significant reduction in the data collection
cost.

%We address the problem of data collection in sensor networks using predictive modeling. Prior
%work has suggested several approaches to capture and exploit the rich spatio-temporal
%correlations prevalent in sensor networks during data collection.
%Although shown to be effective in reducing the data collection cost, many of these
%approaches use a simplistic correlations model and further, ignore many
%idiosyncrasies of sensor networks, in particular the broadcast nature of communication.
%%Our proposed approach is based on approximating the joint probability distribution
%%over the sensors using {\em undirected graphical models}, ideally suited to
%%exploit both the spatial correlations and the broadcast nature of communication.
%Our proposed approach is based on the notion of ``compression tree'' in which
%the data of sensor nodes are predicted by their parents.
%The approach only involves using pairwise correlations, thus is more practical.
%We present algorithms for %optimally
%using such a model for data collection under
%different communication models, and for identifying an appropriate model to use for a given
%sensor network. We develop centralized approximation algorithms that
%work for both wireless and wired networks.
%Experiments over synthetic and real-world datasets show that our
%approach significantly reduces the data collection cost.
\end{abstract}

\section{Introduction} % and Problem Definition}
\noindent{In} this paper, we address the problem of designing energy-efficient protocols
for collecting all data observed by the sensor nodes in a sensor
network at an Internet-connected base station, at a specified
frequency.
The key challenges in designing an energy-efficient data collection protocol are
%modeling and exploiting the strong spatio-temporal correlations present in most sensor networks (see Figure 1).
effectively exploiting the strong spatio-temporal correlations present in most sensor networks, and
optimizing the routing plan for data movement.
In most sensor network deployments, especially
in environmental monitoring applications, the data generated by the sensor nodes
is highly correlated both in time (future values are correlated with
current values) and in space (two co-located sensors are strongly correlated).
%In the naive protocol, data from each source is simply
%sent to the base station through a shortest path; however this
%tends to be significantly suboptimal
%This renders a total data transmission cost $\sum_{X_i} H(X_i)\dist(X_i,BS)$
%where $\dist(i,BS)$ is the length of the shortest path to the base station.
%since it completely ignores the fact that the joint entropy of the nodes is much
%less than the sum of the individual entropies because of strong spatial correlations.
%Naive data collection protocols tend to be significantly suboptimal in the presence of such correlations.
These correlations can usually be captured
by constructing predictive models using either prior domain knowledge, or
historical data traces.
However, the distributed nature of data generation and the resource-constrained nature of the
sensor devices, make it a challenge to optimally exploit these correlations.
    %exploiting the correlations is made challenging because of the distributed nature
%of data generation, and because of the resource-constrained nature of the sensor nodes.
%However, optimally using them is challenging for several reasons, and hence this problem
%has seen much work in recent years.

%We note that capturing and expressing correlation between many nodes
%are usually computation- and storage- expensive. Therefore, we focus
%on making use of only pairwise correlations
%for data compression such that
%the total communication cost is minimized.
%Moreover, the distributed nature of sensor networks
%poses new challenges of judiciously deciding the data movement scheme.
%Besides the issue of practicality,
%considering pairwise correlations also gives rise to the notion of ``compression tree''
%which allows us to use powerful techniques from combinatorial optimization.
%generation in sensor networks, and the resource-constrained nature of sensor
%nodes, traditional data compression techniques cannot be easily adapted to
%exploit such correlations.

Consider an $n$-node sensor network, with node $i$ monitoring the value of a variable $X_i$, and
generating a data flow at entropy rate of $H(X_i)$.
%We assume a data flow is generated at each node $i$ at some entropy rate of $H(X_i), i = 1, ..., n$,
%the base station needs to collect the data from all nodes.
In the naive protocol, data from each source is simply
sent to the base station through the shortest path, rendering
a total data transmission cost $\sum_{i} H(X_i) \cdot \dist(i,BS)$,
where $\dist(i,BS)$ is the length of a shortest path to the base station.
However, because of the strong spatial correlations among the $X_i$, the joint entropy of the
nodes, $H(X_1, \ldots, X_n)$, is typically much smaller than the sum of the individual entropies; the
naive protocol ignores these correlations.
%This however ignores the strong spatial correlations between the nodes $X_i$; because of these correlations, the joint entropy of the nodes, $H(X_1, \dots, X_n)$, is
%typically much less than the sum of the individual entropies.
%distributive data is usually highly correlated, i.e, the joint entropy is much less
%than the sum of individual entropies.

A lower bound on the total number of bits that need to be communicated
can be computed using the {\em Distributed Source Coding (DSC) theorem}~\cite{SW1973,wyner,xiong,journals/tosn/Su07}.
In their seminal work, Slepian and Wolf~\cite{SW1973} prove that it is theoretically possible to encode
the correlated information generated by distributed data sources (in our case,
the sensor nodes) at the rate of their joint  entropy {\em even if  the data
sources do not communicate with each other}.
This can be translated into the following lower bound on the total amount of data transmitted for
a multi-hop network:
$\sum_{i} \dist(i, BS) \times H(X_i | X_1, \dots, X_{i-1})$
where $X_1,\ldots,X_n$ are sorted in an increasing
order by their distances to the base station~\cite{CBV2004,journals/tosn/Su07}.
With high spatial correlation, this number is expected to be much smaller than the
total cost for the naive protocol (i.e., $H(X_i | X_1, \dots, X_{i-1}) \ll H(X_i)$).
%We note this could be much lower than the cost $\sum_{X_i} \dist(X_i, BS)\times H(X_i)$ incurred
%by the naive solution.
The DSC result unfortunately is non-constructive, with constructive techniques known for only a few
specific distributions~\cite{pradhan}; more importantly, DSC requires perfect knowledge of the correlations
among the nodes, and may return wrong answers if the observed data values deviate from what is expected.
%non-constructive, and constructive techniques are known only for a few specific
%distributions~\cite{pradhan}. More importantly, these techniques require
%precise and perfect knowledge of the correlations.  This may not be acceptable
%in practical sensor networks where deviations from the modeled correlations must
%be captured accurately (we use DSC to provide a lower bound on the data
%collection cost; see Section \ref{sec:prior-approaches}).

However, the lower bound does suggest that significant
savings in total cost are possible by exploiting the correlations.
Pattem et al.~\cite{PKG04}, Chu et al.~\cite{CDHH06}, Cristescu et al.~\cite{cristescu}, among others, propose
practical data collection protocols that exploit the spatio-temporal correlations while
guaranteeing correctness (through explicit communication among the sensor nodes). These protocols may exploit only a subset of the correlations, and
in many cases, assume uniform entropies and conditional entropies. Further, most of this prior
work has not attempted to provide any approximation guarantees on the solutions, nor have they attempted a
rigorous analysis of how the performance of the proposed data collection protocol compares with the lower
bound suggested by DSC.
%however, these protocols may exploit only a subset of
%the correlations. These previous approaches however either used simplistic network models
%(assuming uniform entropies etc.), or largely did not provide guarantees on the solution. They also
%did not work with a specific probability distribution, instead assuming uniform entropies and
%well-defined expressions for the conditional entropies.

We are interested in understanding how to get as close to the DSC lower bound as possible for a given
sensor network and a given set of correlations among the sensor nodes.
In a recent work, Liu et al.~\cite{conf/mobicom/Liu06} considered a similar problem to ours and developed
an algorithm that performs very well compared to the DSC lower bound. However, their results are implicitly
based on the assumption that the conditional entropies are quite substantial compared to the base variable entropies
(specifically, that $H(X_i | X_1, ..., X_{i-1})$ is lower bounded). Our results here are complimentary in that,
we specifically target the case when the conditional entropies are close to zero (i.e., the correlations
are strong), and we are able to obtain approximation algorithms for that case. We note that we are also
able to prove that obtaining better approximation guarantees is NP-hard, so our results are tight for that
case. As we will see later, lower bounding conditional entropies enables us to get better approximation
results and further exploration of this remains a rich area of future work.

In this paper, we analyze the data collection problem under the restriction that any data collection protocol can directly
utilize only {\em second-order marginal or conditional} probability distributions -- in other words, we only
directly utilize pair-wise correlations between the sensor nodes.
There are several reasons for studying this problem. First off, the entropy function typically obeys a
strong diminishing returns property in that, utilizing higher-order distributions may not yield
significant benefits over using only second-order distributions. Second, learning, and utilizing,
second-order distributions is much easier than learning higher-order distributions (which can typically
require very high volumes of training data). Finally, we can theoretically analyze the problem of finding the optimal data collection scheme
under this restriction,
and we are able to develop polynomial-time approximation algorithms for solving it.
%(under one model, we obtain a polynomial time algorithm to solve the problem optimally).
%to obtain approximation algorithms for solving it (under one model, a polynomial time algorithm).
%I put this in to remind us that we need to fix it.
%this typically does not restrict us in any significant way. There are three reasons to do this:
%(1) we can analyze the problem very well, and develop approximation algorithms for it; (2)
%learning such models is easier since we need to learn few parameters, and (3) compression
%tasks at the sensor nodes are easier since we only need to store two variable probability
%distributions there.

This restriction leads to what we call {\em compression trees}.
Generally speaking, a compression tree is simply a directed spanning tree $\calT$ of the communication network
in which the parents are used to compress the values of the children. More specifically, given a directed edge
$(u, v)$ in $\calT$, the value of $X_v$ is compressed using the value of $X_u$\footnote{In the rest of the
paper, we denote this by $X_v|X_u$} (i.e., we use the value of $X_u = x_u$ to compute the conditional distribution
$p(X_v | X_u = x_u)$ and use this distribution to compress the observed value of $X_v$ (using say Huffman coding)).
The compression tree also specifies a data movement scheme, specifying where (at which sensor node) and how the values of $X_u$ and $X_v$ are
collected for compression.
%Consider a directed edge $(u, v)$ in $\calT$.
%This indicates that the
%The raw data of $X$ and $Y$ should be received by some node
%where the conditional data $Y|X$ (of $H(Y|X)$ bits) will be computed and sent to the base station.
%It can be shown that the base station can restore the data of each node by induction on the tree structure.

\begin{figure*}
    \hspace{-10pt}
    {\includegraphics[width=7.4in]{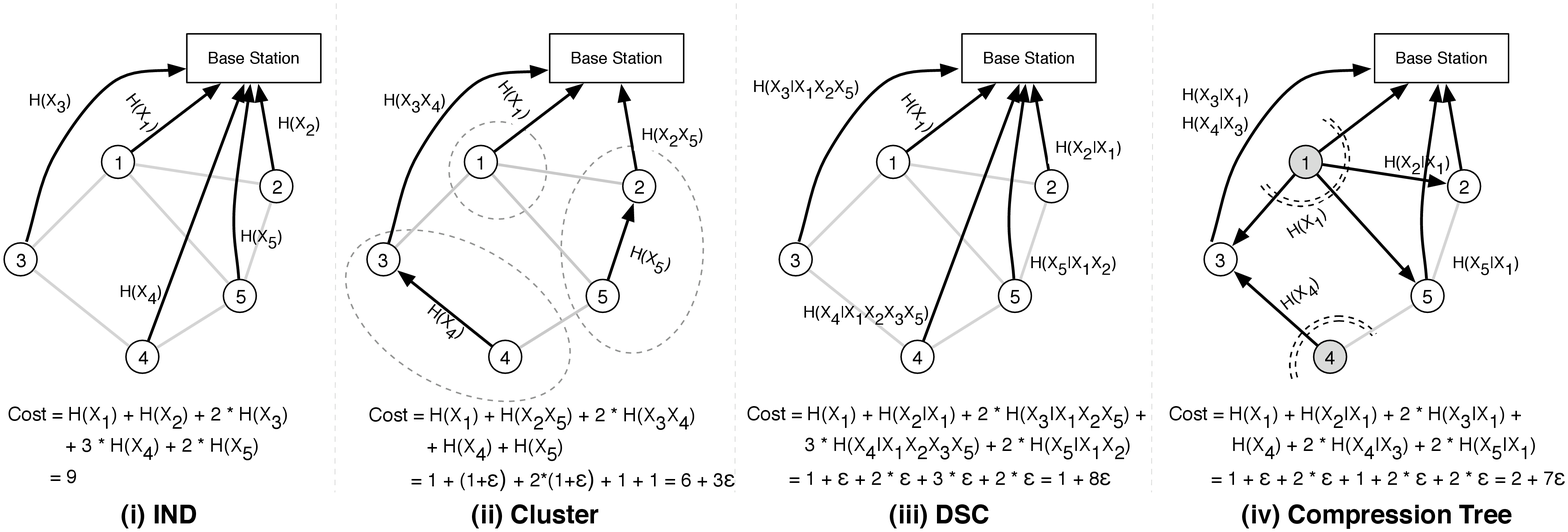}}
    \vspace{-20pt}
    \caption{Illustrating different data collection approaches -- costs computed assuming $H(X_i) = 1, H(X_i|X_j) = \epsilon, \forall~i,~j$:
    (i) IND: correlations ignored; (ii) Cluster: using 3 clusters $\{X_1\}, \{X_2, X_5\}, \{X_3, X_4\}$; (iii) DSC (theoretical optimal);  (iv) Compression tree: with edges $1 \rightarrow 2$, $1 \rightarrow 3$, $1 \rightarrow 5$
    and $3 \rightarrow 4$ (the cost under WN model would have been $5 + 7\epsilon$).
    }
    \label{fig:prior-approaches}
\end{figure*}

The compression tree-based approach can be seen as a special case of the approach
presented by one of the authors in prior work~\cite{wang}. There the authors
proposed using {\em decomposable models} for data collection in wireless sensor
networks, of which compression trees can be seen as a special case. However,
that work only presented heuristics for solving the problem, and did not present
any rigorous analysis or approximation guarantees.

%One of the promising approaches to exploit the spatial correlations optimally in
%Our compression tree based scheme is mainly inspired by
%Wang and Deshpande's low-complexity prediction model\cite{}.
%More concretely, they use the {\em undirected graphical model} to approximate
%the joint distribution over the sensor nodes and construct the {\em junction tree}
%which specifies how the predictions are conducted.
%We note that our compression tree model in which only pairwise correlations are used for prediction
%is a special case of the more general graphical model.
%Before the formal definition,
%we explain the way the compression works through the example.
%Also see Figure \ref{fig:example-prediction-tree}.

%Given these, we are asked to find a compression tree over the sensor nodes and a
%communication {\em scheme} indicating where the compression should be performed.

\section{Problem Definition}
\label{sec:problemdefinition}
We begin by presenting preliminary background on data compression in
sensor networks, discuss the prior approaches, and then introduce
the compression tree-based approach.

%\begin{figure}[t]
        %\includegraphics[width=2.5in]{bit-vs-messages}
        %\vspace{-12pt}
        %\caption{Two extremes in the spectrum of communication model and data encoding options}
        %\label{fig:correlations}
%\end{figure}

\subsection{Notation and Preliminaries}
We are given a sensor network modeled as an undirected, edge-weighted graph $\calG_C(V = \{1, \cdots, n\},E)$,
comprising of $n$ nodes that are continuously
monitoring a set of distributed attributes $\calX = \{X_1, \cdots, X_n\}$.
%The edge set $E$ consists of pairs of vertices that are within communication radius of each other.
The edge set $E$ consists of pairs of vertices that are within communication radius of each other,
with the edge weights denoting the communication costs.
%and generate a discrete data value vector ${\bf x}^t = \{x_1^t, \cdots, x_n^t\}$ at every
%time instance $t$\footnote{The time instances at which data is acquired depends
%on the application-specified frequency of data collection.}.
Each attribute, $X_i$, observed by node $i$, may be an environmental
property being sensed by the node (e.g., {\em temperature}), or it may be the
result of an operation on the sensed values (e.g., in an anomaly-detection
application, the sensor node may continuously evaluate a filter such as
``$temp > 100$'' on the observed values).
If the sensed attributes are continuous, we assume that an error threshold
of $e$ is provided and
the readings are binned into intervals of size $2e$ to discretize them.
In this paper, we focus on optimal exploitation of spatial correlations at any
given time $t$; our approach can be generalized to handle temporal correlations
in a straightforward manner.
%we briefly discuss how the algorithms can be generalized to handle temporal
%correlations in a straightforward manner.

%More formally, we are given a sensor network with $n$ nodes that continuously
%monitors a set of distributed attributes $\calX = \{X_1, \cdots, X_n\}$.
We are also provided with the entropy rate for each attribute, $H(X_i)~(1\leq i\leq n)$ and
the conditional entropy rates, $H(X_i|H_j)~(1\leq i,j\leq n)$, over all pairs of attributes.
More generally, we may be provided with a joint probability distribution, $p(X_1, ..., X_n)$,
over the attributes, using which we can compute the joint entropy rate for any subset of
attributes. However accurate computation of such joint entropies for large subsets of attributes is usually not feasible.
%More generally, we may also be provided with the joint entropy rates for any subset of attributes, $H(X_{i_1}, \cdots, X_{i_k}$);
%\footnote{I think this should be redefined to be a function that provides
%the value of entropy given a set of sensor nodes, since thats what we really need}.
%We denote the communication graph of the sensor network by $\calG_C = (\calX, E)$,
%with edge weights denoting the cost of sending one bit across that edge.
%We allow for asymmetric communication links.
%We consider both the broadcast communication
%model (in which case the edge weights are uniform) and point-to-point
%communication model.

%Predictive modeling-based approaches to data compression begin by building a predictive model over
%the sensor network attributes that is used to obtain a joint probability distribution (pdf) over
%the attributes. We denote this pdf by $p(X_1, ..., X_n)$.

%We denote the communication graph of the sensor network by $\calG_C = (\calX, E)$, where $E$
%consists of the pairs of vertices that are within communication radius of each other.

We denote the set of neighbors of the node $i$ by $N(i)$ and let
$\bar{N}(i)=N(i)\cup\{i\}$ and $\deg(i)=|N(i)|$. We denote by
$\dist(i, j)$ the energy cost of communicating one bit of
information along the shortest path between $i$ and $j$.

\choosesecondifkeepwired{
We focus on the wireless communication model (WL) in this paper; specifically we assume that when a
node transmits a message, all its neighbors can hear the message ({\em broadcast} model). We
further assume that the energy cost of receiving such a broadcast message is negligible,
and we only count the cost of transmitting the message.
In the extended version of the paper, we discuss how our approach generalizes to wired communication
networks, and to unicast or multicast models.
}{
We consider the following communication cost models in this paper. The data
movement schemes and how the costs are counted differ among
different models.
%specified later.
\begin{descriptionsmallermargin}
    \item[Wireless Network (WL)] In this model, when a node transmits a message,
    all its neighbors can hear the message ({\em broadcast} model). We further assume that
    the energy cost of receiving such a broadcast message is negligible, and we only count the cost of transmitting
    the message.
    If the unicast protocol is used, the network behaves as a wired network (see below).
    %For several sensor node devices, the cost of receiving a message at a sensor
    %node can be non-trivial (perhaps even as high as the {\em transmission} cost); we assume that the
    %network behaves as a wired network in that
%    In that case, we assume that the network behaves as a wired network (see below). %\footnote{This
        %typically results in
    %We briefly discuss what happens if the receiving cost is the same as
        %the transmission cost (we can use the wired network solution with an additional factor of 2.
        %Some similar bounds exist otherwise.

\item[Wired Network(WN)] Here we assume point-to-point communication without any
broadcast functionality. Each communication link is
weighted, denoting the cost of transmitting one bit of message
through this link.
%We further make a distinction between {\em multicast} communication, which allows for sharing of transmissions
%when a sender needs to communicate a piece of information to multiple receivers, and
%{\em unicast} communication, where that is not allowed. %Under the multicast model, if a node $i$ needs to send a message to
%We assume the communication links to be weighted, denoting the costs of transmitting the messages.
    \begin{descriptionsmallermargin}
    \item[Multicast]
    When a sender needs to communicate a piece of information to multiple receivers,
    we allow for sharing of transmissions. Namely, a message can be sent from the source
     to a set of terminals through a Steiner tree.
    \item[Unicast]
    Each communication is between two node (one sender and one receiver).
    Different message transmissions cannot be shared and
    the cost of each communication is counted separately.
    %Each communication happens between the source node and the destination.
    \end{descriptionsmallermargin}
%both nodes $j$ and $k$, the communication may be shared ; the unicast model does not allow this type of sharing.
%
%Under unicast communication, a message sent from $i$ to $j$ is
%may be relayed by other sensor nodes, but different message transmissions cannot be {\em shared}, even
%if the contents are identical. If multicast is allowed, then the messages may be shared
%A node can send a message to another node at a certain cost. However, it cannot broadcast
%to a set of neighbors at a fixed cost (messages have to be sent to each
%neighbor individually).
\end{descriptionsmallermargin}
}

\subsection{Prior Approaches} %Predictive Modeling-based Data Compression in Sensor Networks}
\label{sec:prior-approaches}
\label{sec:dsc}
Given the entropy and the joint entropy rates for compressing the sensor network attributes,
the key issue with using them for data compression is that the values are generated
in a distributed fashion. The naive approach to using {\em all} the correlations in the
data is to gather the sensed values at a central sensor node, and compress them
jointly. However, even if the compression itself
was feasible, the data gathering cost would typically dwarf any advantages gained
by doing joint compression.
Prior research in this area has suggested several approaches that utilize a subset
of correlations instead. Several of these approaches are illustrated in Figure \ref{fig:prior-approaches}
using a simple 5-node sensor network.
\begin{descriptionsmallermargin}
    \item[IND] Each node compresses its own value, and sends it to the base station along the
    shortest path. The total communication cost is given by $\sum_i \dist(i,BS) \cdot H(X_i)$.
%One approach, called {\em Independent (IND)}, is to ignore the spatial
%correlations and to compress the data from each sensor node independently of the
%others (Figure \ref{fig:prior-approaches} (i)). In other words, an approximate
%distribution $q_1(X_1, ..., X_n) = p(X_1) p(X_2) ... p(X_n)$ is used for compression
%(where $p(X_i)$ denotes the marginal probability
%distribution of $X_i$, computed by summing over the remaining variables in $\calX$).
\item[Cluster] In this approach~\cite{PKG04,CDHH06}, the sensor nodes are
grouped into clusters, and the data from the nodes in each cluster
is gathered at a node (which may be different for different clusters) and is compressed jointly.
Figure \ref{fig:prior-approaches} (ii) shows an example of this using
three clusters $\{1\}, \{2, 5\}, \{3, 4\}$. Thus
the intra-cluster spatial correlations are exploited during compression; however, the correlations
across clusters are not utilized.

\item[Cristescu et al.~\cite{cristescu}] The approach proposed by Cristescu et al. is similar to ours,
    and also only uses second-order distributions.
    \choosesecondifkeepwired{
       However they only consider the unicast communication model, and further assume that the
       entropies and conditional entropies are uniform.
       Rickenbach et al.~\cite{Rickenbach04gatheringcorrelated} also present results under similar
       assumptions.
    }{
    They present algorithms for the WN case, further assuming
    that the entropies and conditional entropies are uniform. The solution space that we consider in this paper
    is larger that the one they consider, in that it allows more freedom in choosing the compression trees; in
    spite of that we are able to develop
    a PTIME algorithm for the problem they address (see Section \ref{sec_unicast})\footnote{However, they require
    all the communication to be along a tree; we don't require that from our solutions.}. Further, we make no
    uniformity assumptions about the entropies or the conditional entropies in that algorithm.
    }

\item[DSC] Distributed source coding (DSC), although not feasible in this setting for the reasons
discussed earlier, can be used to obtain a lower bound on total communication
cost as follows~\cite{CBV2004,cristescu,journals/tosn/Su07}. Let the sensor nodes be numbered
in increasing order by distances from the base station (i.e., for all $i$,
$\dist(i, BS) \le \dist(i+1, BS)$). The optimal scheme for using DSC is as follows:
$X_1$ is compressed by itself, and
transmitted directly to the sink (incurring a total cost of $\dist(1, BS) \times
H(X_1)$). Then, $X_2$ is compressed according to the conditional distribution of $X_2$
given the value of $X_1$,  resulting in a data flow rate of $H(X_2 | X_1)$ (since the sink
already has the value of $X_1$, it is able to decode according this
distribution). Note that, according to the distributed source coding
theorem~\cite{SW1973},
sensor node $2$ does not need to know the actual value of $X_1$. Similarly, $X_i$ is
compressed according to its conditional distribution given the values of $X_1, \dots, X_{i-1}$.
The total communication cost incurred by this scheme is given by: \\[2pt]
\centerline{$\sum_{i=1}^n \dist(i, BS) \times H(X_i | X_1, \dots, X_{i-1})$}
Figure \ref{fig:prior-approaches} (iii) shows this for our running example (note that $5$ is
closer to sink than $3$ or $4$).

%Forcing DSC to only use pairwise correlations is somewhat non-trivial however. We need to compute
%the compression tree that minimizes a complex function. We can do this by adding directed edges from
%each node to each other node etc., and computing the minimum cost
%{\em branching} (directed spanning
%tree)~\cite{xxx}.
\item[RDC] Several approaches where data is compressed along the way to the base station ({\em routing driven compression}~\cite{PKG04,SS02,GE03})
have also been suggested. These however require joint compression and decompression of large
numbers of data sources inside the network, and hence may not be suitable for resource-constrained sensor networks.

\item[Dominating Set-based] Kotidis~\cite{conf/icde/Kotidis05} and Gupta et al.~\cite{conf/mobihoc/GuptaNDC05}, among others,
    consider approaches based on using a representative set of sensor nodes to approximate the data distribution over the
entire network; these approaches however do not solve the problem of exact data
collection, and cannot provide correctness guarantees.
\end{descriptionsmallermargin}

\noindent{As} we can see in Figure \ref{fig:prior-approaches}, if the spatial correlation is high, both IND and Cluster
incur much higher communication costs than DSC. For example, if $H(X_i) = 1,
\forall i$, and if $H(X_i | X_j)  = \epsilon \approx 0, \forall i, j$ (i.e., if the spatial
correlations are almost perfect), the total communication costs of IND, Cluster
(as shown in the figure), and DSC would be $9, 6$, and $1$ respectively.

%Several other approaches based on {\em routing driven compression}~\cite{PKG04,SS02} have also been suggested. However,
%these approaches typically require joint compression and decompression of large
%numbers of data sources inside the network, and may not be practical given the resource constraints on the sensor
%devices.

%\begin{figure*}[t]
    %\centerline{\includegraphics[width=5.5in]{prior-approaches}}
    %\caption{Illustration of three prior approaches to data commpression for a
    %5-node network (CLSTR uses 3 clusters $\{X_1\}, \{X_2, X_5\}, \{X_3,
    %X_4\}$). If spatial correlations are perfect, total communication
    %costs (using the {\em bit-hop} metric) for IND and \cluster can be very high compared to the
    %theoretical optimal DSC.}
    %\label{fig:prior-approaches}
%\end{figure*}

\subsection{Compression Trees}
\eat{
In practice, we are likely to be limited to using only low-order marginal or conditional probability
distributions for compression in sensor networks. There are several reasons for this. In most cases,
a collection of interacting lower-order probability distributions are usually expressive enough to
capture complex correlations among a large number of variables (e.g., using a Bayesian or a Markov network).
The entropy function is known to be submodular, and typically exhibits a strong diminishing returns property.
Second, it is typically not possible to learn accurate joint probability distributions over
more than a few variables -- that not requires large training datasets, but is also prone to
overfitting. Finally, compression using higher-order joint distributions not only entails higher communication
among the sensor nodes, but it may not be feasible on the typically resource-constrained sensor devices.
}

%In the real application, deriving the correlation between two datasets usually requires
%some node gathers both datasets and conducts some learning algorithms (**not sure).
%Furthermore, deriving the correlation between many (even three) nodes
%is often computationally expensive and expressing it needs significant amount of storage.
%Therefore, for practical purposes, the following constraints are imposed:
%\begin{enumerate}
%\item
%We can make use of only pairwise correlations.
%\item
%If two nodes are compressed jointly, then the data from both nodes must be gathered at a single location.
%\end{enumerate}

As discussed in the introduction, in practice, we are likely to be limited to using only low-order marginal or conditional probability
distributions for compression in sensor networks.
In this paper, we begin a formal analysis of such algorithms by analyzing the problem of optimally exploiting the spatial correlations under
the restriction that we can only use second-order conditional distributions (i.e., two-variable
probability distributions).
%We now formally define the {\em compression tree} model.
A feasible solution under this restriction is fully specified by a
directed spanning tree $\calT$ rooted at $r$ (called a {\em compression tree})
and a data movement scheme according to $\calT$.
In particular, the compression tree indicates which of the second-order distributions are to be used,
and the data movement scheme specifies an actual plan to implement it.

More formally, let $p(i)$ denote the parent of $i$ in $\calT$. This indicates that both
$X_i$ and $X_{p(i)}$ should be gathered together at some common sensor node, and that $X_i$
should be compressed using its conditional probability distribution given the value
of $X_{p(i)}$ (i.e., $p(X_i | X_{p(i)} = x_{p(i)})$). The compressed value is communicated to the
base station along the shortest path, resulting in an entropy rate of $H(X_i | X_{p(i)})$.
Finally, the root of the tree, $r$, sends it own value directly to the base station, resulting
in an entropy rate of $H(X_r)$. It is easy to see that the base station can reconstruct
all the values. The data movement plan specifies how the values of $X_i$ and $X_{p(i)}$ are
collected together for all $i$.

In this paper, we address the optimization problem of finding the optimal compression tree
that minimizes the total communication cost, for a given communication
topology and a given probability distribution over the sensor network variables (or the entropy rates for all variables, and
the joint entropy rates for all pairs of variables).

We note that the notion of compression trees is quite similar to the so-called {\em Chow-Liu trees}~\cite{chowliu},
used for approximating large joint probability distributions.

%*** WILL THIS STAY?***
\vspace{2pt}
\begin{example}
Figure \ref{fig:prior-approaches} (iv) shows the process of collecting data using a compression tree
for our running example, under the broadcast communication model.
The compression tree (not explicitly shown) consists of four edges: $1 \rightarrow 2$, $1 \rightarrow 3$, $1 \rightarrow 5$ and $3 \rightarrow 4$.
The data collections steps are:
\begin{myenumerate}
\item Sensor nodes $1$ and $4$ broadcast their values, using $H(X_1)$ and $H(X_4)$ bits
respectively. The Base Station receives the value of $X_1$ in this step.
\item Sensor nodes $2$, $3$, and $5$ receive the value of $X_1$, and compress their own values
using the conditional distributions given $X_1$. Each of them sends the compressed
values to the base station along the shortest path.
\item Sensor node $3$ also receives the value of $X_4$, and it compresses $X_4$ using its
own value. It sends the compressed value (at an entropy rate of $H(X_4|X_3)$) to the base
station along the shortest path.
\end{myenumerate}
The total (expected) communication cost is thus given by: \\[2pt]
        \centerline{$H(X_1) + H(X_4) + H(X_2 | X_1) + 2 \times H(X_3 | X_1) +$}\\[2pt]
        \centerline{~~~~~~~~~~~~~~~~~~~~~~~~~$2 \times H(X_5 | X_1) + 2 \times H(X_4 | X_3)$}

%Sensor node $X_1$ (root) compresses its information, using the
%distribution $p(X_1)$, independently and transmits it to the base station.
%item $X_2$ is compressed using the conditional probability distribution $p(X_2
%| X_1)$. This requires that the observed values $X_1$ and $X_2$ be collected
%at some sensor node. This could be done at any sensor node in the network. In
%the example it is done at node $X_2$ (value of $X_1$ is transmitted to $X_2$)
%%\item Similarly, $X_3$ and $X_4$ are compressed using $p(X_3|X_1)$ and
%%$p(X_4|X_1)$ respectively at sensor nodes $3$ and $4$.
%\item Finally, $X_5$ is compressed using $p(X_5|X_3)$. In this case, the
%compression is done at node $3$.
%\end{enumerate}

%It is easy to see that the expected total number of bits received by the base station is:
%\[ H(X_1) + H(X_2 | X_1) + H(X_3 | X_1) + H(X_4 | X_1) + H(X_5 | X_3) \]

%The minimum number of bits that the base station must receive is simply:
%\[ H(X_1, X_2, X_3, X_4, X_5) \]

%Under the {\em broadcast} model, assuming receiving is free, and assuming uniform
%links, the total communication cost can be seen to be (note
%that when data is sent by $X_3$ and $X_4$ it takes two broadcasts
%for it to reach the Base Station):
%\begin{eqnarray*}
%&& H(X_1) + H(X_5) + H(X_2 | X_1) + 2 \times H(X_3 | X_1) +  \\
%&& 2 \times H(X_4 | X_1) + 2 \times H(X_5 | X_3)
%\end{eqnarray*}

If the conditional entropies are very low, as is usually the case, the total
cost will be simply $H(X_1) + H(X_4)$.
\end{example}

%On the other hand, the cost of the naive approach which does not exploit correlations would be:
%\[ H(X_1) + H(X_2) + 2 \times H(X_3) + 2 \times H(X_4) + 2 \times H(X_5) \]

%Note that we could have used any compression tree, the correctness is not
%compromised.

\subsection{Compression Quality of a Solution} %Structure of the Solution and Data Compression Quality}
To analyze and compare the quality of the solutions with the DSC approach,
   we subdivide the total communication cost incurred by a data
   collection approach into two parts:
%The total communication cost incurred by any data collection appraoch
%can be subdivided into two parts:
%Let us first look at the structure of a feasible solution under our model.
%In general, such a solution can be fully specified by a compression tree $\calT$
%and a data movement scheme $\calM_\calT$ according to $\calT$.
%The data movement scheme $\calM_\calT$ consists of two parts, intro-source communication
%and the conditional data communication.
%We briefly discuss the structure of these two parts and how they attribute
%to the gap between the cost of our compression tree technique and the DSC lower bound.

\begin{descriptionsmallermargin}
\item[Necessary Communication (NC)] %$CC(\calM_\calT)$.
As discussed above, for practical reasons, data collection schemes typically use a
subset of the correlations present in the data (e.g. Cluster only uses intra-cluster
correlations, our approach only uses second-order joint distributions). Given the
specific set of correlations utilized by an approach, there is a minimum amount
of communication that will be incurred during data collection. This cost
is obtained by computing the DSC cost assuming only those correlations are present in the
data. For a specific compression tree, the NC cost is computed as: \\[2pt]
\centerline{$H(X_r) \times \dist(r, BS) + \sum_{i\in V} H(X_i|X_{p(i)})\times \dist(i, BS)$}

The NC cost for the Cluster solution shown in Figure \ref{fig:prior-approaches}(ii)
is $4 + 5 \epsilon$, computed as: \\[2pt]
\centerline{\small $H(X_1) + H(X_2) + 2 \cdot H(X_5|X_2) + 2 \cdot H(X_3) + 3 \cdot H(X_4 | X_3)$}

In some sense, NC cost measures the penalty of ignoring some of the correlations during
compression. For Cluster, this is typically quite high -- compare to the NC cost for DSC ($= 1 + 8 \epsilon$).
On the other hand, the NC cost for the solution in Figure \ref{fig:prior-approaches} (iv) is
$1 + 8\epsilon$ (i.e., it is equal to the NC cost of DSC -- we note that this is an
artifact of having uniform conditional entropies, and does not always hold).

\item[Intra-source Communication (IC)] This measures the cost of explicitly gathering the
data together as required for joint compression. By definition, this cost is 0 for DSC.
We compute this by subtracting the NC cost from the total cost.
For the solutions presented in Figures \ref{fig:prior-approaches} (ii) and (iv),
the IC cost is $2 - 2\epsilon$ and $1-\epsilon$ respectively. The broadcast communication model
significantly helps in reducing this cost for our approach.

\end{descriptionsmallermargin}

\noindent{The} key advantage of our compression tree-based approach is that its NC cost is usually quite close to
DSC, whereas the other approaches, such as Cluster, can have very high NC costs because they ignore
a large portion of the correlations.
%We say that a data movement scheme that
%satisfies the above condition {\em implements} the compression tree $\calT$.
%When we refer to a solution for the problem, we refer to
%a compression tree and a data movement scheme implementing it and
%the cost of the solution is the cost of the data movement scheme.

%By increasing the expressive power of the model used and thus capturing larger subsets of spatial correlations (for example,
%by increasing the cluster sizes), we can reduce the Approximation Loss, but the increase in the Intra-source Communication cost
%will typically outweigh the benefits (e.g. in Ken~\cite{CDHH06}, the optimal cluster sizes were found to be $< 4$).

%\subsection{Different Cost Models}
\subsection{Solution Space}
In our optimization algorithms, we consider searching among two different classes of compression trees.
\begin{mylist}
\item {\em Subgraphs of $\calG$ (SG):} Here we require that the compression tree be a subgraph
of the communication graph. In other words, we compress $X_i$ using $X_j$ only if $i$ and $j$
are neighbors.
\item {\em No restrictions (NS):} Here we don't put any restrictions on the compression trees.
As expected, searching through this solution space is much harder than SG.
\end{mylist}
In general, we expect to find the optimal solution in the SG solution space; this is because
the correlations are likely to be stronger among neighboring sensor nodes than among sensor
nodes that are far away from each other.

%Combined with our earlier distinction between different communication models, we get
%four different problems that we consider in the next two sections: (1) WN-SG, (2) WL-SG,
%(3) WL-NS, and (4) WN-NS.
%Now, we briefly discuss the differences between the various models.
%\begin{enumerate}
%\item {\bf WL-SG:} Wireless network and $\calT$ is a SubGraph of $\calG$. We require
%$T_i\subseteq \bar N_{\calG}(i)$ for all $i\in V(\calG_C)$.
%If $T_i\ne \{i\}$, then $c(T_i)=H(X_i)$, i.e, $i$ has to broadcast $X_i$.
%\item {\bf WL-NS:} WireLess network and $\calT$ is not Necessarily a Subgraph of $\calG$.
%In this case, $T_i$ can be any connected set of nodes.
%$c(T_i)$ is the number of nodes that broadcast so that all nodes in $T_i$ get $X_i$.
%We can see the minimum value is the cardinality of a minimum connected dominating
%set covering $T_i$.
%\item {\bf WN-SG:} Wired Network and $\calT$ is a SubGraph of $\calG$.
%$T_i\subseteq \bar N_{\calG}(i)$ for each node $X_i$ and\\
%$c(T_i)=H(X_i)\sum\limits_{j\in T_i\setminus\{i\}}c(i,j)$.
%\item {\bf WN-NS:} Wired Network and $\calT$ is not Necessarily a Subgraph of $\calG$.
%In this case, $T_i$ can be any connected set of nodes.
%$c(T_i)$ is the cost of the minimum tree spanning all nodes in $T_i$.
%\end{enumerate}

Finally, we define $\beta$
as the {\em bounded conditional entropy parameter}, which
bounds the ratio of conditional entropies for any pair of variables
that can be used to compress each other.
Formally, ${1\over \beta}\leq {H(X_i|X_j)\over H(X_j|X_i)}\leq \beta$
for any nodes $i$ and $j$ and some constant $\beta\geq 1$.
For the $SG$ problem, this is taken over pairs of adjacent nodes
and for the $NS$ problem, it is taken over all pairs.
Moreover, the above property
implies that the ratio of entropies between any pair of nodes is also bounded,
${1\over \beta}\leq {H(X_i)\over H(X_j)}\leq \beta$.

%Note that if all nodes have uniform entropy, then $\beta = 1$
%($H(X_i) = H(X_j) \Rightarrow H(X_i | X_j) = H(X_j | X_i)$).

%Note that $\beta=1$ if the entropy of all nodes is uniform
%because $H(X_i,X_j)=H(X_i)+H(X_j|X_i)=H(X_j)+H(X_j|X_i)$.

%\red{should this be defined only for adjacent pairs of nodes or any ? otherwise we can't
%say the following statement.}
We expect $\beta$ to be quite small ($\approx 1$) in most cases (especially
if we restrict our search space to SG). Note that, if the entropies are uniform ($H(X_i) = H(X_j)$), then $\beta = 1$.
%($H(X_i) = H(X_j) \Rightarrow H(X_i | X_j) = H(X_j | X_i)$).

\subsection{Summary of Our Results}
\choosesecondifkeepwired{
We refer to the two problems that we focus on in this paper by WL-SG (where compression trees are
restricted to be subgraphs of $\calG$), and WL-NS (no restrictions on compression trees).
Below we summarize our key results.
}{
Combining the distinct communication models and different solution
spaces, we get four different problems that we consider in this
paper: (1) WL-SG, (2) WL-NS, (3) Multicast-NS, and (4) Unicast
(which subsumes Multicast-SG). We summarize the results as follows.
}

\begin{enumerate}
\item(Section \ref{subsec_wcds})
    We first consider the WL-SG problem under an {\em uniform entropy and conditional entropy assumption},
    i.e.,  we assume that $H(X_i) = 1~\forall i$ and $H(X_i | X_j) = \epsilon ~\forall i, j, i \ne j$.
    We develop
%We consider the WL-SG problem under uniform entropy and conditional entropy assumption.
%If we require the compression tree is a subgraph of the communication graph,
a $\left({1\over 1+2\epsilon (\davg-1/2)}(H_\Delta+1)+2
\right)$-approximation for this problem,
where $\davg$ is the average distance to the base station. % and $\epsilon$ is ratio of
%the conditional entropy to the entropy.
\item(Section \ref{subsec_framework} and \ref{subsec_analysis})
We develop a unified generic greedy framework which can be used for
approximating the problem under various communication cost models.
\item(Section \ref{subsec_wlsg} and \ref{subsec_wlns})
    We show that, for wireless communication model,
%In wireless network model,
the greedy framework gives a $4\beta^2 H_n$ approximation factor
%if the compression tree is required to be a subgraph of the communication graph
for the SG solution space and and an $O(\beta^3 n^\epsilon\log n)$
(for any $\epsilon>0$) factor for the NS solution space.
\choosesecondifkeepwired{}{
\item(Section \ref{subsec_wire_notsubgraph})
%In wired networks with multicast communication model,
For multicast-NS problem, we show that the greedy framework gives an
$O({\beta^3\over \epsilon}(\log n)^{3+\epsilon})$ (for any
$\epsilon>0$) approximation.
\item(Section \ref{sec_unicast})
For the unicast communication model, we present a simple poly-time algorithm for finding
an optimal {\em restricted} solution (defined in Section \ref{subsec_framework}), giving us a $(2+\beta)$-approximation.
Further, we show that the optimal restricted solution is also the optimal solution under uniform
entropy and conditional entropies assumption.
}
\item (Section \ref{sec:experiment}) We illustrate through an empirical evaluation
that our approach usually leads to very good
data collection schemes in presence of strong correlations. In many
cases, the solution found by our approach performs nearly as well as
the theoretical lower bound given by DSC.
\end{enumerate}

\section{Approximation Algorithms}
\label{sec:approxalgo}
We first present an approximation algorithm for the WL-SG problem under
the uniform entropy assumption; this will help us tie the problem with some
previously studied graph problems, and will also form the basis for our main algorithms.
We then present a generic greedy framework that we use to derive approximation
algorithms for the remaining problems.

\ignore{
\subsection{The Structure of the Solutions}

In general, a feasible solution is fully specified by a compression tree $\calT$
and a data movement scheme $\calM_\calT$ according to $\calT$.
The compression tree $\calT$ is simply a tree rooted at the base station ($BS$),
spanning all vertices of $\calG$.
Note that $\calT$ may or may not be a subgraph of $\calG$ depending
on users' requirement.
The data movement scheme $\calM_\calT$ consists of two parts as follows:
\begin{enumerate}
\item Raw information movement $RI(\calM_\calT)$.
Let $T_{i}$ be the set of nodes that receive the raw information $X_i$ of node $i$.
In all our models, $T_i$ is connected and includes node $i$.
However, the cost $c(T_{i})$ for sending $X_i$ from $i$ to all nodes in $T_{i}$ differs
in different models.
\item Conditional information movement $CI(\calM_\calT)$.
Node $i$ being the parent of node $j$ in $\calT$ means that $T_{i}\cap T_{j}\ne \emptyset$
and  we compress $X_j$ by conditioning it on $X_i$
at some node in the intersection, and then send the conditional data
$X_j|X_i$ to $BS$ along a shortest path.
Let $I_{i,j}$ be the place where $X_j|X_i$ is computed.
\end{enumerate}
We say that a data movement scheme that
satisfies the above condition {\em implements} the compression tree $\calT$.
When we refer to a solution for the problem, we refer to
a compression tree and a data movement scheme implementing it and
the cost of the solution is the cost of the data movement scheme.

}

%\color{black}
\subsection{The WL-SG Model: Uniform Entropy and Conditional Entropy Assumption}
\label{subsec_wcds}
Without loss of generality, we assume that $H(X_i) = 1,\ \ \forall i$ and
$H(X_i|X_j) = \epsilon \ \ \forall i, j$, for all adjacent pairs of nodes
$(X_i, X_j)$. We expect that typically $\epsilon \ll 1$.

For any compression tree that satisfies the SG property, the
data movement scheme must have a subset of the sensor nodes
locally broadcast their (compressed) values, such that
for every edge $(u, v)$ in the compression tree, either $u$ or $v$
(or both) broadcast their values. (If this is not true, then it is
not possible to compress $X_v$ using $X_u$.) Let $S$ denote this subset of nodes.
%Along with this we also need
%the root of the compression tree to transmit its value to the
%base station.
Each of the remaining nodes only transmits $\epsilon$
bits of information. % -- thus the cost of the local broadcasts dominates
%the total cost of the protocol.

%Our goal is to select a compression tree and a communication
%scheme so that the base station can reconstruct all the values.
%Since each edge in the compression tree must be a
%In this case, our goal is to select a subset $S\subset V$ of
%nodes that perform a local broadcast of their values to their
%neighbors. All remaining nodes will send only conditional entropy
%information to the base station.

%Let $S$ denote the subset of sensor nodes that broadcast their values locally.
To ensure that the base station can reconstruct all the values, $S$ must further
satisfy the following properties:
(1) $S$ must form a dominating set of $\calG_C$ (any node $\notin S$ must have a neighbor in $S$).
(2) The graph formed by deleting all edges $(x,y)$ where
$x,y \in V \setminus S$ is connected.
Property (1) implies every node should get at least one of its neighbors' message for compression
and property (2) guarantees the connectedness of the  compression tree given $S$ broadcast.
Graph-theoretically this leads to a slightly different problem
than both the classical Dominating Set (DS) and Connected Dominating Set
(CDS) problems\cite{journals/algorithmica/GuhaK98}.
Specifically, $S$ must be a
Weakly Connected Dominating Set (WCDS)~\cite{Mobihoc2002_chenliestman} of $\calG_C$.

In the network shown in Figure~\ref{fig:prediction-example},
nodes $4, 3, 9$ and $10$ form a WCDS, and thus locally broadcasting
them can give us a valid compression tree (shown in Figure \ref{fig:prediction-example} (ii)).  However, note that
nodes $4, 9, 10$ and $2$ form a DS but not a WCDS. As a result, we cannot
form a compression tree with these nodes performing local broadcasts (there
would be no way to reconstruct the value of both $X_3$ and $X_2$).
%node 3 can send $X_3|X_2$ but $BS$ will not have $X_2$.
%Fig.~\ref{fig:prediction-example} shows a network with the
%compression tree shown. Nodes $4, 3, 9$ and $10$ form a
%WCDS.  However, note that nodes $4,9, 10$ and $2$ form a DS
%but not a WCDS. As a result, we cannot form a compression tree
%with these 4 nodes performing a local broadcast.

% as well as the information about which node
%predicts which node and where the prediction is done.

\begin{figure*}
\begin{minipage}{0.48\linewidth}
{
{\includegraphics[scale=0.45]{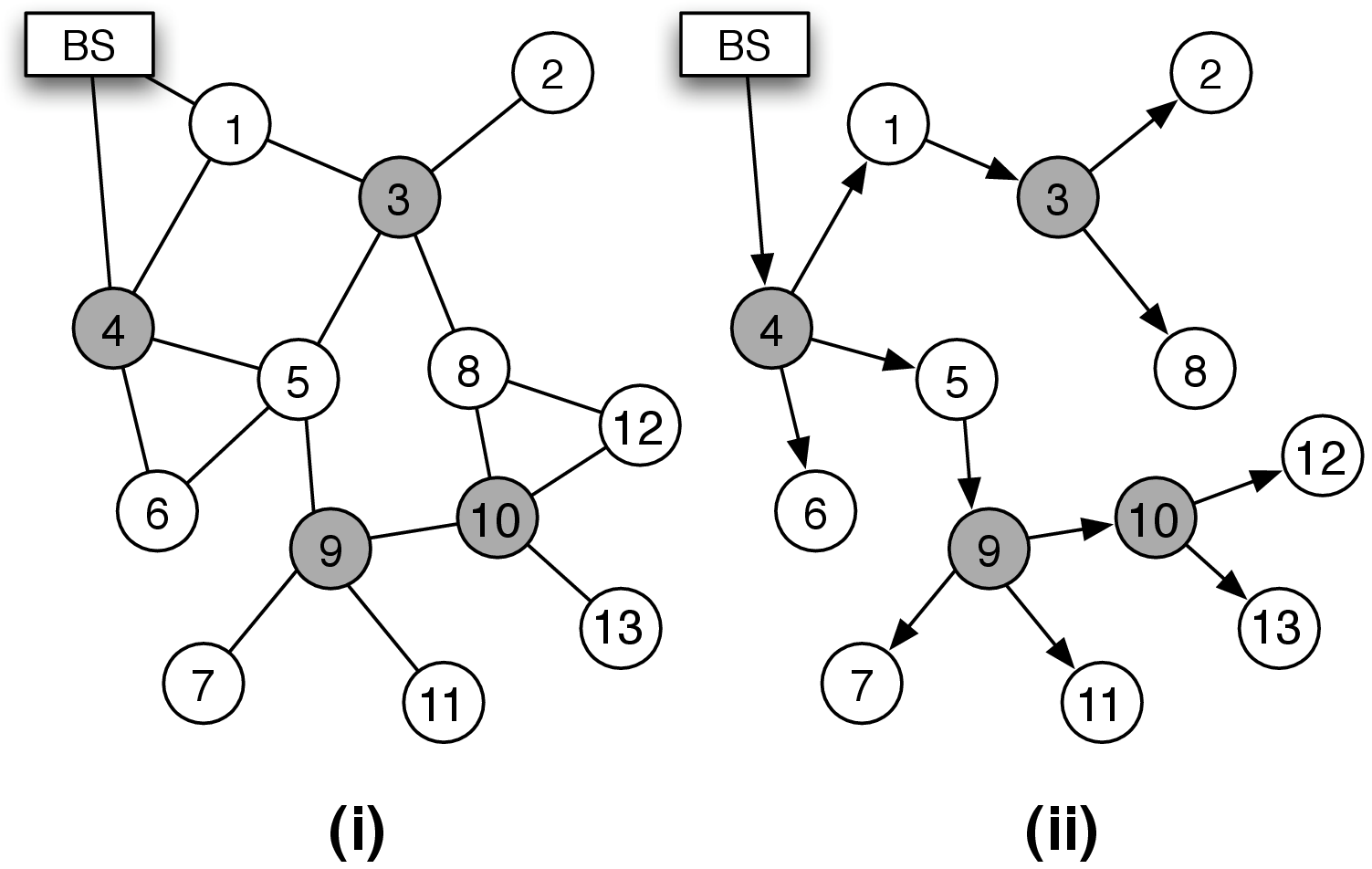}}
\caption{(i) A weakly connected dominating set of the sensor network is
indicated by the shaded nodes, which locally broadcast their values; (ii) The corresponding
compression tree (e.g. Node 3 is compressed using the value of Node 1 at Node 1, whereas
Node 5 is compressed using the value of Node 4 at Node 5).}
\label{fig:prediction-example}
}
\end{minipage}
~ ~ ~
\begin{minipage}{0.48\linewidth}
{
{\includegraphics[scale=0.45]{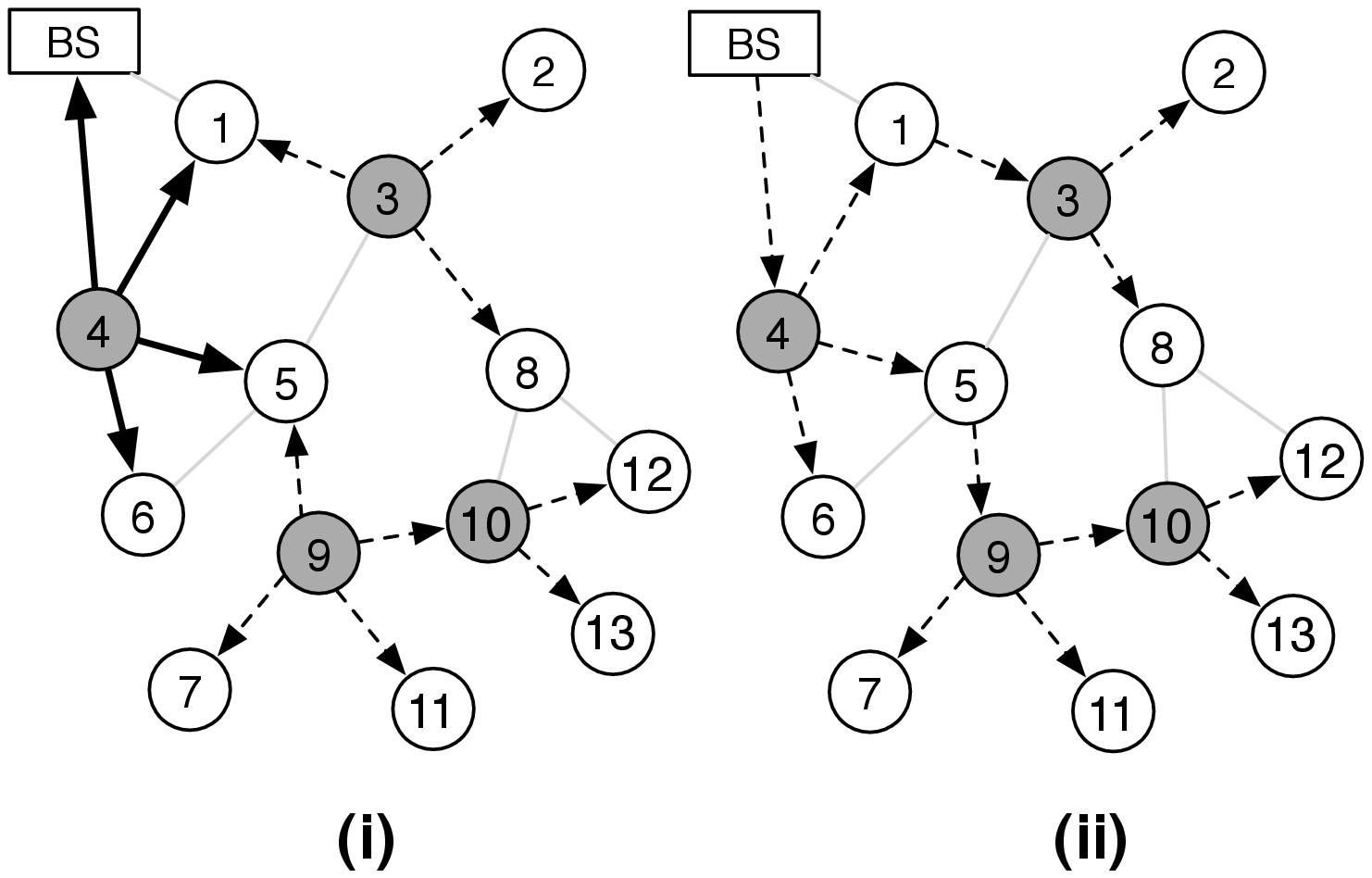}}
    \caption{Illustrating the Treestar algorithm: First the treestars centered at
nodes 10, 9 and 3 are chosen, and finally the treestar centered at node
4 is chosen. This causes the parents of nodes 1 and 5 to be re-defined as node 4,
the parent of node 9 to be defined as node 5, and the parent of node 3 to
be defined as node 1. (i) also shows an extended compression tree.}
\label{fig:greedy2}
}
\end{minipage}

\end{figure*}

The approach for the CDS problem that gives a $2H_\Delta$
approximation \cite{journals/algorithmica/GuhaK98}, gives a
$H_\Delta+1$ approximation\footnote{$\Delta$ is the maximum
degree and $H_n$ is the $n$th harmonic number, i.e,  $H_n=\sum_{i=1}^n {1\over i}$.}
for WCDS \cite{Mobihoc2002_chenliestman}. We use this to
prove that:

\iffalse
\subsubsection{A Greedy Algorithm} The idea behind the algorithm is
the following: we will choose the set $S$ by growing a ``tree'' $T$.

Initially all vertices are  unmarked (white). Initially we select a vertex
of maximum degree (color it black) and add it to $S$.
We mark all its neighbors that are not in $T$ and add them to $T$ (color
them gray). Thus gray nodes are adjacent to chosen (black) nodes.
To select a new node to add to $S$, we can either select a gray node or select
a white node that is adjacent to a gray node and color it black. Each
time we color a node black, we add its white neighbors to the tree $T$
and color them gray.

This algorithm will end up choosing  a
weakly connected dominating set ($CDS$) in the end.

The main question is the following: what rule should we use for
picking a vertex to add to $S$? A natural choice is to pick the
vertex that has the maximum number of unmarked (white) neighbors. We call
this the ``yield'' of the choice.
\fi

\begin{Theorem}
Let the average distance to the base station be $\davg={\sum_j \dist(j,BS)\over n}$.
The approximation for WCSD yields a $\left({1\over 1+2\epsilon (\davg-1/2)}(H_\Delta+1)+2 \right)$-approximation
for WL-SG problem under uniform entropy and conditional entropies assumption.
\end{Theorem}

\omitproofs{
\begin{proof}
The amount of data broadcast is clearly $|S|$ (since $H(X_i) = 1$
for $i \in S$). Each non-broadcast node $j$ sends $\epsilon$ amount of data
to BS -- the cost of this is $\epsilon \dist(j,BS)$.
For each broadcast node $j$, $\epsilon$ amount of data may be sent from $p(j)$
and the cost is at most $\epsilon (\dist(j,BS)+1)$ and at least
$\epsilon (\dist(j,BS)-1)$.
The total communication cost is thus at most $UB=|S|+\epsilon (\sum_j \dist(j,BS)+|S|)$.
In an optimal solution, suppose $S^*$ denotes the set of nodes that perform local broadcast;
then the lower bound on the total cost is: $LB=|S^*| + \epsilon (\sum_j d(j,BS)-|S^*|)$.
%Taking the ratio yields the following theorem.
We can also easily see $|S^*|\leq n/2$, thus
$|S^*|\leq {1\over 1+2\epsilon(\davg-1/2)}LB$.
Therefore,
%$
%UB
%\leq (H(\Delta)+1)|S^*| + \epsilon (d_{avg}+1)n
%\leq (H(\Delta)+1)|S^*| + 2\epsilon \left(d_{avg}n-|S^*|\right)
%\leq (H(\Delta)+1){1\over 1+\epsilon (d_{avg}n-1/2)}LB + 2\epsilon LB
%\leq {1\over 1+\epsilon (d_{avg}-1/2)}(H(\Delta)+1)+2 LB.
%$
\begin{eqnarray*}
UB
&\leq& (H_\Delta+1)|S^*| + \epsilon (\davg+1)n
\leq  (H_\Delta+1)|S^*| + 2\epsilon \davg n \\
&\leq& {(H_\Delta+1) \over 1+2\epsilon (\davg -1/2)}LB + 2 LB
\leq \left({1\over 1+2\epsilon (d_{avg}-1/2)}(H_\Delta+1)+2
\right) LB.
\end{eqnarray*}
\qed
\end{proof}
}

From the above theorem, if $\epsilon$ is small enough, say $\epsilon=o({1\over \davg})$,
the approximation ratio is approximately $H_\Delta)$.
On the other hand, if $\epsilon$ is large, the approximation ratio becomes better.
Specifically, if $\epsilon \approx H_\Delta/\davg$, then we get a constant approximation.
This matches our intuition that the hardness of approximation comes mainly from
the case when the correlations are very strong. We can further formalize this --
by a standard reduction from the set cover problem
which is hard to approximate within a factor of $(1 - \delta)\ln n$
%for which it is hard to approximate within in a factor of $(1-\delta)\ln n$
for any $\delta>0$~\cite{feige98}, we can prove:

\begin{Theorem}
The WL-SG problem can not be approximated
within a factor of $(1-\delta)\ln n$ for any $\delta>0$ even with uniform entropy and
conditional entropy, unless $NP\subseteq DTIME(n^{\log\log n})$.
\end{Theorem}

\subsection{The Generic Greedy Framework}
\label{subsec_framework}
We next present a generic greedy framework that helps us analyze the rest of
the problems.

%\subsubsection{The Restricted Solution and the Extended Compression Tree}

Suppose node $p(i)$ is the parent of node $i$ in the compression tree $\calT$.
%If compression can only happen between neighbors, node $I_{i,j}$,
%where $X_j|X_i$ is computed, is either $i$ or $j$.
%However in the other two cases, i.e, the WL-NS model and the WN-NS model,
Let $I_{i,p(i)}$ denote the node where $X_i$ is compressed using $X_{p(i)}$.
We note that this is not required to be $i$ or $j$, and could be any node
in the network. This makes the analysis of the algorithms very hard. Hence
we focus on the set of feasible solutions of the following restricted form:
$I_{i,j}$ is either node $i$ or $j$. The following lemma states that the cost of the optimal restricted solution
is close to the optimal cost.

\begin{Lemma}
\label{lm_wire_restrict} Let the optimal solution be $\opt$ and the
optimal restricted solution be $\ropt$. We have $\cost(\ropt)\leq
(2+\beta)\cost(\opt)$. Furthermore, for WL-SG model,
$\cost(\ropt)\leq 2\cost(\opt)$.
\end{Lemma}

\omitproofs{
\begin{proof}
%We only prove the lemma for the wired network model.
%The proof of the wireless model are essentially the same.
Let $\calT^*$ be the compression tree of $\opt$.
We keep the compression tree unchanged and only modify the data movement scheme
$\opt$ to construct a restricted solution $\ropt$ whose cost is at most
$(2+\beta)\cost(\opt)$.
%Suppose $\calT$ is the optimal compression tree.
Assume that $i$ is the parent of $j$ in $\calT^*$ and $X_j|X_i$ is computed at some node $I_{i,j}$.
We denote by $T_i$ the set of nodes which receive $X_i$ from $i$.
We simply extend $T_{i}$ to be $\tilde{T}_{i}=T_{i}\cup P(I_{i,j},j)$ where
$P(u,v)$ is the shortest path from $u$ to $v$. Then $X_j|X_i$ is computed on
node $j$ and then sent to the base station.

Now, we analyze the increase in cost for the wired network model.
The proof for the wireless network case is almost the same and we omit it here.
Let $p(i)$ be the parent and $Ch(i)$ be the set of children of node $i$ in $\calT^*$.

\begin{eqnarray*}
\cost(\opt)&=& \sum_{i\in \calT^*}H(X_i)c(T_i)
+\sum_{i\in \calT^*\setminus\{BS\}}H(X_i|X_{p(i)})\dist(I_{i,p(i)},BS).
\end{eqnarray*}

Thus, we have:

\begin{eqnarray*}
 \cost(\ropt)  &=& \sum_{i\in \calT^*}H(X_i)c(\tilde{T}_{i})+\sum_{i\in \calT^*\setminus\{BS\}}H(X_i|X_{p(i)})\dist(i,BS) \\
          &\leq& \sum_{i\in \calT^*}H(X_i)\left(c(T_i)+\sum_{j\in Ch(i)} \dist(I_{j,i},j)\right)
               +\sum_{i\in \calT^*\setminus\{BS\}}H(X_i|X_{p(i)})
                 \left(\dist(I_{i,p(i)},i)+\dist(I_{i,p(i)},BS)\right)\\
          &\le&    \cost(\opt)+\beta \sum_{i\in \calT^*}H(X_i)\dist(I_{i,p(i)},i)
              + \sum_{i \in \calT^*} H(X_i|X_{p(i)}) \dist(I_{i,p(i)},i) \\
          &\leq& \cost(\opt)+(1+\beta)\sum_{i\in \calT^*}c(T_i) \leq (2+\beta)\cost(\opt).
\end{eqnarray*}

For the WL-SG model, the only reason that $I_{i,p(i)}$ is neither $i$ nor $p(i)$
is that both $i$ and $p(i)$ broadcast their values to the third node $I_{i,p(i)}$
which is closer to the base station.
The above analysis can be still carried over except we don't need any extra intra-source communication.
Then, we don't have the $\beta$ term in the formula and it gives us a ratio of $2$.
\qed

\end{proof}
}

%\section{The Greedy Algorithm for the WL-SG, WL-NS, WN-NS Models}
%\subsection{The Generic Greedy Framework}

%All previous solutions to this problem were  two phases methods\cite{YYYY}.
%A compression tree was designed in the first phase and a corresponding
%data movement scheme was subsequently computed for that compression tree.
%Our algorithm combines these two steps together.

\iffalse
Our generic greedy framework uses an approach similar to the
one used in the algorithms for approximating Node Weighted Steiner Trees~\cite{journals/acm/KlRa95,journals/acm/GuKh99}
and Spanning Tree with Inner Node cost~\cite{pdcat05_msti}.
We note our algorithm will find a restricted solution.
\fi

Our algorithm finds what we call an {\em extended compression tree}, which in a final
step is converted to a compression tree. An {\em extended compression tree} $\vcalT$ corresponding to
a compression tree $\calT$ has the same underlying tree structure,
%tree structure as the compression tree $\calT$,
but each edge $e(i,j)\in \calT$ is associated with an
{\em orientation} specifying the raw data movement.
Basically, an extended compression tree naturally suggests a
restricted solution in which an edge from $i$ to $j$ in $\vcalT$
implies that $i$ ships its raw data to $j$ and the corresponding
compression is carried out at $j$. %, and
%that the corresponding compression is done on $j$.
We note that the direction of the edges in $\vcalT$ may not be the same as
in $\calT$ where edges are always oriented from the root to the leaves,
irrespective of the data movement.
In the following, we refer {\em the parent of } node $i$ to be the parent in $\calT$,
i.e, the node one hop closer to the root, denoted by $p(i)$.

The main algorithm greedily constructs an extended compression tree
by greedily choosing subtrees to merge in iterations.
We start with a empty graph $\calF_1$ that consists of only isolated nodes.
During the execution, we maintain a forest in which each edge is directed.
In each iteration, we combine some trees together into a new larger tree
by choosing the most (or approximately) cost-effective
{\em treestar} (defined later).
Let the forest at the start of the $i$th iteration be $\calF_i$.
A treestar $\TS$ is specified by $k$ trees in $\calF_i$,
say $T_1,\ldots,T_k$, a node $r\notin T_j (1\leq j\leq k)$ and $k$ directed edges $e_j=(r,v_j) (v_j\in T_j,1\leq j\leq k)$
%We call $T_0$ the {\em center-tree} ,
We call $r$ the center, $T_1,\ldots,T_k$ the {\em leaf-trees}, $e_j$ the {\em leaf-edges}.
The treestar $\TS$ is a specification of the data movement of $X_r$, which
we will explain in detail shortly.
Once a treestar is chosen, the corresponding data movement is
added to our solution.
The algorithm terminates when only one tree is left which will be
our extended compression tree $\vcalT$.

Let $r$ be the center of $\TS$ and $S$ be the subset indices of leaf-trees.
We define the cost of $\TS$ ($\cost(\TS)$) to be
$$
\min_{v_j\in T_j, j\in S} ( c(r,\{v_j\}_{j\in S})H(X_r)  + \sum_{j\in S} H(X_{v_j}|X_r)\dist(v_j,BS))
$$
where $c(r,\{v_j\}_{j\in S})$ is the minimum cost for sending $X_r$ from $r$ to all $v_j$'s.
Essentially, the first term corresponds roughly to the cost of
intra-source communication (raw data movement of $X_r$), denoted $IC(\TS)$
and the second roughly to the necessary communication (conditional data movement), denoted $NC(\TS)$.
We say that the corresponding data movement is an {\em implementation} of the treestar.
The cost function $c()$ differs
for different cost models of the problem; we will specify its concrete form later.

We define the {\em cost effectiveness} of the treestar $\TS$
to be $ceff(\TS)={\cost(\TS)\over k+1}$ where $k$ is the number of leaf-trees in $\TS$.
In each iteration, we will try to find the most cost effective
treestar.
%We will discuss how to find such treestar in the actual discussion of each model.
Let {\em Mce-Treestar($\calF_i$)} be the procedure for finding the
most (or approximately) cost effective treestar on $\calF_i$.
The actual implementation of the procedure {\em Mce-Treestar} will be described in detail
in the discussion of each cost model.
In some cases, finding the most cost-effective treestar
is NP-hard and we can only approximate it.

\eat{
\begin{figure}[t]
    \vspace{-5pt}
    %\centerline{\includegraphics[width=2.5in]{greedy2}}
    \centerline{\includegraphics[width=3.6in]{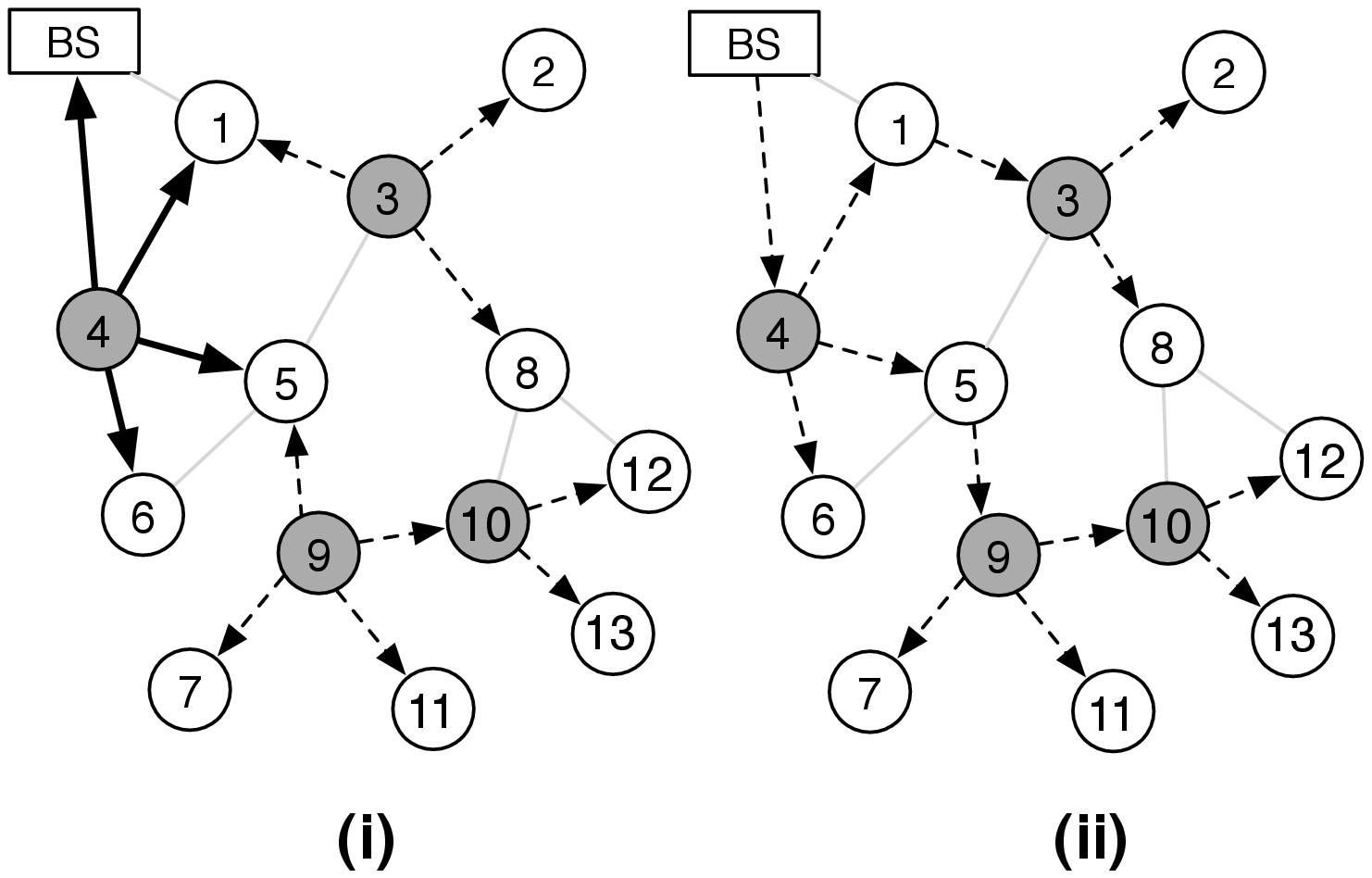}}
    \vspace{-5pt}
    \caption{Illustrating the Treestar algorithm: First the treestars centered at
nodes 10, 9 and 3 are chosen, and finally the treestar centered at node
4 is chosen. This causes the parents of nodes 1 and 5 to be re-defined as node 4,
the parent of node 9 to be defined as node 5, and the parent of node 3 to
be defined as node 1.}
    \label{fig:greedy2}
\end{figure}
}

We now discuss the final data movement
scheme and how the cost of the final solution has been properly accounted
%up to a factor of $\beta$
in the treestars that were chosen.
Suppose in some iteration, a treestar $\TS$ is chosen in which the center node $r$
sends its raw information to each $v_j (v_j\in T_j, j\in S)$ ($S$
is the set of indices of leaf-trees in $\TS$).
%We call it the raw information communication {\em induced by $\TS$} and denote it by $R(\TS)$.
The definition of the cost function suggests that
%the raw information $H(X)$ be sent from $X$ to all $X_j$s and
$X_{v_j}$ is compressed using $X_r$ at $v_j$,
and the result is sent from $v_j$ to BS.
However, this may not be consistent with the
extended compression tree $\vcalT$.
In other words, some $v_j$ may later become the parent of $r$, due to latter
treestars being chosen,
in $\vcalT$ which implies that $r$ should be compressed using $v_j$
instead of the other way around.
Suppose some leaf $v_p (v_p\in T_p, p\in S)$ is the parent of $r$ in $\vcalT$.
The actual data movement scheme is determined as follows.
We keep the raw data movement induced by $\TS$ unchanged, i.e, $r$ still sends $X_r$ to each $v_j (j\in S)$.
But now, $X_r|X_{v_p}$ instead of $X_{v_p}|X_r$ is computed
on node $v_p$ and sent to the base station.
Other leaves $v_j (j\ne p)$ still compute and send $X_{v_j}|X_r$.
It is easy to check this data movement scheme
actually implements the extended compression tree $\vcalT$.

For instance, in Figure~\ref{fig:greedy2}, %note that in the compression tree,
node 3 is initially the parent of node 1, but later node 4 becomes the
parent of node 1, and in fact node 1 ships $X_1|X_4$ to the base
station (and not $X_1|X_3$). Node 1 now being the parent of node 3
also compresses $X_3$ and sends $X_3|X_1$ to $BS$.
%in the compression tree is used to compress 3. However, the raw data transmission of
%$X_3$ is done from
%node 3 to nodes 1, 2 and 8. Node 1 sends $X_3|X_1$ to  $BS$.
Due to the fact that ${1\over \beta}\leq {H(X|Y)\over H(Y|X)}\leq \beta$,
the actual data movement cost is at most $\beta$ times the sum of the treestar costs.
Thus every part of the communication cost incurred is counted in some treestar.
%We have found and every part of the communication cost is counted in
%some treestar.
We formalize the above observations as the following lemma:
%We will specify the actual cost function $c()$ in the detailed discussion of each model.
%However, from the above discussion, we can see
%the cost of the final solution is properly counted in these treestars.
%We formalize it as the following lemma.
\begin{Lemma}
Let $\TS_i$ be the treestars we choose in iteration $i$ for $1\leq i\leq \ell$. Then:
$\cost(\calT)\leq \beta\sum_{i=1}^\ell \cost(\TS_i) $.
\end{Lemma}

The pseudocode for constructing $\vcalT$ and the corresponding
communication scheme is given in Algorithm \ref{alg_generic}.

%*** redefine parents***

%\linesnumbered
\begin{algorithm}[t]
\caption{The Generic Greedy Framework}
\label{alg_generic}
 $\calF_1=\bigcup_{i=1}^n \{\{X_i\}\}$\;
 $ i \rightarrow 1$\;
 \While{ $\calF_i$ is not a spanning tree} {
       $\TS_i=Mce-Treestar(\calF_i)$\;
       Let $E(\TS_i)$ be tree-edges of $TS_i$ and $r$ is the center of $\TS_i$\;
       $\calF_{i+1}\leftarrow \calF_i+E(\TS_i)$\;
       $T_r\leftarrow T_r+IC(\TS_i)$\;
       $i=i+1$;
 }
 $\vcalT=\calF_i$\;
 \For{each directed edge $e(i,j)\in E(\vcalT)$}{
       \eIf{$i$ is the parent of $j$} {
              Compute $X_j|X_i$ at $j$ and send it to $BS$;
       }{
              Compute $X_i|X_j$ at $j$ and send it to $BS$;
       }
 }
\end{algorithm}

%However, at the first glance, the cost of a treestar explicitly uses the fact
%that the center will compress the leaf-trees.
%The following lemma states that the cost of treestars actually add up to the cost
%of our final solution even the parent-children relationship may get changed.

\subsection{The Generic Analysis Framework}
\label{subsec_analysis}
Let $\calF_i$ be the forest of
$n_i$ trees before iteration $i$ and $\vcalT$ be the final extended compression tree.
$OPT$ is defined as the optimal restricted solution and $OPT_i$ as the optimal solution
for the following problem:
Find a extended compression tree that contains $\calF_i$ as a subgraph
such that the cost for implementing all treestars in $\vcalT-F_i$ is minimized.
Clearly, $OPT_1=OPT$.
Let $\TS_i$ be the treestar computed in iteration $i$, with $m_i$ tree
components (the number of leaf-trees of $\TS_i$ plus one).
After $\ell$ iterations (it is easy to see
$\ell$ must be smaller than $n$), the algorithm terminates.
It is easy to see $n_{i+1}=n_i-m_i+1$ for $i=1,\ldots,\ell-1$.
We assume {\em Mce-Treestar} is guaranteed to find an
$\alpha$-approximate most cost-effective treestar.

\begin{Lemma}
For all $i\ge 1$, ${cost(\TS_i)\over m_i} \le \alpha{cost(OPT_i)\over n_i}$.
\end{Lemma}
\omitproofs{
\begin{proof}
Suppose the extended compression tree for $OPT_i$ is $\vcalT_i$
that has $\calF_i$ as a subgraph.
$OPT_i$ consists of all data movement which implements
all treestars defined by $\vcalT_i-\calF_i$.
These treestars, say $\TS^*_1,\TS^*_2,\ldots$,
correspond to edge disjoint stars in $\vcalT_i$.
Suppose $\TS^*_j$ connects $m^*_j$ tree components.
Since each tree component of $\calF_i$ is involved in some $\TS^*_i$,
we can see $\sum_j m_j^*\geq n_i$.
By the fact that $\TS_i$ is a $\alpha$-approximation of the most effective treestar,
we can get
\begin{eqnarray*}
{\cost(\TS_i)\over m_i} &\leq& \alpha\min_{j}\left\{{\cost(\TS^*_j)\over m^*_j}\right\}
\leq \alpha{\sum_j\cost(\TS^*_j)\over \sum_{j} m^*_j} \\
&\leq& \alpha{\cost(OPT_i)\over n_i}.
\end{eqnarray*}
\qed
\end{proof}
}

%The following lemma states that the cost to connect $\calF_i$ is at most the one to connect $\calF_1$ for $i>1$.
The proof of the following lemma is omitted.
\begin{Lemma}
$\cost(OPT_i)\leq \cost(OPT)$.
\end{Lemma}
\iffalse
\begin{proof}
We construct an extended compression tree which has an implementation
with cost at most $\cost(OPT)$.
Let $\vcalT^*$ be the extended compression tree of $OPT$.
We first superimpose $\vcalT^*$ and $\calF_i$.
Then, we repeatedly break cycles by only deleting edges which originally only belong to $OPT$.
Let the resultant tree be $\tilde{\calT}$.
It is easy to see
all raw data movement and
a subset of conditional data movement of $OPT$
suffice to implement all treestars in $\tilde{\calT}-F_i$.
In fact, if $e(i,j)\in \tilde{\calT}-F_i$,
we keep the computation and movement of $X_i|X_j$ (or $X_j|X_i$).
%By superimposing all raw data movement and
%a subset of conditional data movement of $OPT$ on the partial solution corresponding to $F_i$
%we form a restricted feasible solution corresponding to $\tilde{\calT}$.
%It is clear that all raw data movement can implement the communication needed by $\tilde{\calT}$.
%If $e(u,v)\in \calT^*\cap \tilde{\calT}$, we keep the computation and movement of $H(v|u)$.
%By the same treatment as what we have done to resolve the parent-children relation
%between treestars and the final compression tree, we can see
%at most $cost(OPT)$ is needed
%to augment $F_i$ to a feasible solution.
\qed
\end{proof}
\fi
We are now ready to prove our main theorem.

\begin{Theorem}
\label{thm:msti}
Assuming we can compute an $\alpha$ approximation of the most cost-effective treestar
and the bounded conditional entropy parameter is $\beta$,
there is a $2\alpha\beta^2 H_n$ approximate restricted solution.
\end{Theorem}
\omitproofs{
\begin{proof}
The cost of our solution $SOL$ is:
{\small
\begin{align*}
& \cost(SOL)\leq \beta\sum_{i=1}^\ell \cost(\TS_i)\leq \beta\sum_{i=1}^\ell \alpha\beta{\cost(OPT_i)m_i\over n_i} \\
          & \,\,\,\,\,\,\,\leq \alpha\beta^2\cdot \cost(OPT) \cdot \sum_{i=1}^\ell {m_i\over n_i}
          \leq 2\alpha\beta^2 H_n \cost(OPT)
\end{align*}
}
The last inequality holds since: $\sum_{i=1}^\ell {m_i\over n_i} \leq 2\sum_{i=1}^\ell {m_i-1\over n_i} \leq 2H_n$.
%{\small
%\begin{eqnarray*}
%\sum_{i=1}^\ell {m_i\over n_i} &\leq& 2\sum_{i=1}^\ell {m_i-1\over n_i} %\\
                      %\leq 2H_n.
%\end{eqnarray*}
%}
\qed
\end{proof}
}

\subsection{The WL-SG Model}
\label{subsec_wlsg}
\label{treestar}
We first specify the cost function $c(r, \{v_j\}_{j\in S})$ in the wireless sensor network model
where we require the compression tree to be a subgraph of the communication graph
and then give a polynomial time algorithm for
finding the most cost-effective treestar.

Recall $c(r, \{v_j\}_{j\in S})$ is cost of sending $X_r$ from $r$ to all $v_j$'s.
It is easy to see $c(r,\{v_j\}_{j\in S})=H(X_r)$ since we require $v_j$ to be adjacent to $r$
and a single broadcast of $X_r$ from $r$ can accomplish the communication.
The most cost-effective treestar can be computed as follows:
We fix a node $r$ as the center to which all
leaf-edges will connect.
Assume $T_1,T_2,\ldots$ are sorted in a non-increasing order of
$h(r,T_j)=\min_{v\in T_j\cap N(r)}H(X_{v}|X_r)\dist(v,BS)$.
$h(r,T_j)$ captures the minimum cost of sending
the data of some node in $T_j\cap N(r)$ conditioned on $X_r$ to the base station.
The most cost-effective treestar is determined simply by
$$
\min_{k}\left\{H(X_r)+\sum_{j=1}^k h(r,T_j) \over k+1\right\}.
$$

We briefly analyze the running time of the algorithm.
In the
pre-processing step, we need to compute $\dist(v,BS)$ for all $v$ by running the single source shortest path
algorithm which takes $O(n^2)$ time.
In each iteration, for each candidate center $r$, sorting $h(r,T_j)$s needs
$O(\deg(r)\log \deg(r))$ time.
So, the most-effective treestar can be found in $O(|E|\log n)$ time.
Since in each iteration, we merge at least two tree components, hence there are
at most $n$ iterations.
Therefore, the total running time is $O(n|E|\log n)$.
Using Lemma \ref{lm_wire_restrict} and Theorem \ref{thm:msti}, we obtain the following.
\begin{Theorem}
\label{thm:wlsg}
We can compute a $4\beta^2 H_n$-approximation
for the WL-SG model in $O(n|E|\log n)$ time.
\end{Theorem}

%\subsubsection{Special Cases}

%\begin{enumerate}
%\item Case: $H(X_i) = h, \forall i, H(X_i | X_j) = 0, \forall i, j$
%In this case, we reduce the problem to a version of connected set cover/dominating set,
%and we can get very good approximation results.

%\item Case: $H(X_i) = h, \forall i, H(X_i | X_j) = c, \forall i, j$
%The above bound extends with a vanishingly small additive component.
%\end{enumerate}
\subsection{The WL-NS Model}
\label{subsec_bc_notsubgraph}
\label{subsec_wlns}

%We still make the uniform entropy assumption.

Here we don't put any restrictions on the compression trees. Thus,
a source node is able to send the message to a set of nodes through
a Steiner tree and
the cost for sending one bit is the sum of the weights of all inner nodes
of the Steiner tree (due to the broadcasting nature of wireless networks).
In graph theoretic terminology, it
is the cost of the connected dominating set that includes the source node
and dominates all terminals.
Formally, the cost of the treestar $\TS$
with node $r$ as the center and $S$ be the set of indices of the leaf-trees
is defined to be: \\
{
\begin{minipage}{3.4in}
\small
$$
\min_{v_j\in T_j}\left(Cds(r,\{v_j\}_{j\in S})H(X_r)+\sum_{j\in S} H(X_{v_j}|X_r)\dist(v_j,BS)\right)
$$
\end{minipage}
}
where $Cds()$ is the minimum connected set dominating all nodes in its argument.

%We will reduce the problem to the following
%{\em min-density node weighted group steiner (MDNGS) problem}.
%\begin{Definition}
%Given an node weighted graph $G$ and a collection of vertex subset $\{g_i\}$,
%find a set $S$ of vertices such that
%$S$ induces a connected graph.
%The objective is to minimize the density
%$
%\sum_{v\in S}w(v)\over |\{g_i | g_i\cap S\ne \emptyset\}|.
%$
%\end{Definition}

%Now, we show how to find a steiner tree of minimum density.
%We define the distance between the node $v$ and the group $g_j$
%to be the (node weighted) shortest path from $v$ to any node of $g_j$,
%i.e, $\dist(v,g_j)=\min_{u\in g_j}\dist(v,u)$.
%W.l.o.g, we assume $\{g_j\}$s are sorted in an increasing order of
%distance to $v$.
%The optimal density is given by
%$$
%\min_{v\in V,2\leq k\leq n}\{ {w(v)+\sum_{j=1}^k (\dist(v,g_j)-w(v))\over k} \}.
%$$

%Suppose that $X$ will connect $k$ trees.
Next, we discuss how to find the most effective treestar.
We reduce the problem to the following version of
the directed steiner tree problem \cite{journals/99_directsteiner}.
\begin{Definition}
Given a weighted directed graph
$G$, a specified root $r\in V(G)$, an integer $k$ and a set
$X\subseteq V$ of terminals,
the {\em D-Steiner($k,r,X$)} problem asks for a minimum weight
directed tree rooted at $r$ that can reach any $k$ terminals in $X$.
\end{Definition}
It has been shown that the D-Steiner($k,r,X$) problem
can be approximated  within a factor of $O(n^\epsilon)$ for any
fixed $\epsilon>0$ within time $O(n^{O({1\over \epsilon})})$ \cite{journals/99_directsteiner}.

The reduction is as follows. We first fix the center $r$.
Then, we create a undirected node-weighted graph $D$.
The weight of each node is $H(X_r)$.
For each node $v$, we create a copy $v'$ with weight
$w(v') = H(X_v|X_r) \dist(v,BS)$
and add an edge $(u,v')$ for each $u\in \bar{N}(v)$.
For each tree component $T_j$, we create a group
$g_j=\{v'|v\in T_j\}$.
Then, we construct the directed edge-weighted graph .
We replace each undirected edge with two directed edges of opposite directions.
For each group $g_i$, we add one node $t_i$ and edges $(v,t_i)$ for all $v\in g_i$.
The following standard trick will transfer the weight on nodes to directed edges.
For each vertex $v\in V(D)$, we replace it with a directed edge $(v',v'')$ with the same weight
as $w(v)$ such that $v'$ absorbs all incoming edges of $v$ and $v''$ takes all outgoing edges of $v$.
We let all $t''_i$s be the terminals we want to connect.
It is easy to see a directed steiner tree connecting $k$ terminals in the new directed graph
corresponds exactly to a treestar with $k$ leaf-trees.
%(I don't feel the construction is clear to write like this...)

\begin{Theorem}
\label{thm:wlns}We develop an $O(\beta^3 n^\epsilon\log n)$-approximation
for the WL-NS model for any fixed constant $\epsilon>0$ in $O(n^{O({1\over \epsilon})})$ time.
\end{Theorem}

\choosesecondifkeepwired{} {
\subsection{The Multicast-NS Model}
\label{subsec_wire_notsubgraph}

%We still apply our greedy framework to this model.
We consider the wired network model and do not require the
compression tree to be a subgraph of $\calG_C$.
First, we need to provide the concrete form of the cost function $c(r, \{v_j\}_{j\in S})$
i.e., the cost of sending one unit of data from $r$ to all $v_j$'s.
In this model, it is easy to see the cost is
the minimum Steiner tree connecting $r$ and all $v_j$'s.

Suppose node $r$ is the center of the treestar $\TS$ and $S$
is the set of the indices of leaf-trees.
According to the communication model and the general cost definition,
the cost of $\TS$ here is defined as: \\
{
\begin{minipage}{3.4in}
\small
$$
\min_{v_j\in T_j}\left(Stn(r,\{v_j\}_{j\in S})H(X_r)+\sum_{j\in S} H(X_{v_j}|X_r)\dist(v_j,BS)\right)
$$
\end{minipage}
}
where $Stn()$ is the minimum steiner tree connecting all nodes in its argument.
%Basically, the formula says that $r$ sends its raw information to each of $v_j$ through
%a steiner tree.
% and each $X_j$ computes $H(X_j|X)$ and sends it to the base station.

Next we show how to find the most cost-effective treestar.
We first fix the center $r$.
Basically, our task is to find a set $S$ of tree
components such that $ {\cost(\TS)\over k+1} $ is minimized.
%If $\{T_j\}_{j\in S}$ is given, we can easily see it is actually a
%group steiner tree problem \cite{journals/00_groupsteiner,journals/06_groupsteiner2}}.
We will convert this problem to a variant of the group steiner tree problem.
Actually, the following min-density variant
has been considered in order to solve the general group steiner tree problem \cite{journals/06_groupsteiner2}.
\begin{Definition}
Given an undirected graph $G$ and a collection of vertex subsets $\{g_i\}$, find
a tree $T$ in $G$ such that
$\cost(T)\over |\{g_i| g_i\cap T\ne \emptyset\}|$
is minimized.
\end{Definition}

Our reduction to the min-density group Steiner problem works as follows.
For each node $v$, we create a copy $v'$ and add
an edge $(v,v')$ with weight $H(X_v|X_r) \dist(v,BS)$.
For each tree component $T_j$, we define a group $g_j=\{v'|v \in T_j\}$.
It is easy to see the cost of the Steiner tree spanning a set of groups
is exactly the cost of the corresponding treestar.

%However, it is not clear how to adopt the existing algorithms for group steiner tree to the above problem.
%Instead, we reduce the problem to a more general problem,
%the directed steiner tree problem \cite{journals/99_directsteiner}.
%\begin{Definition}
%Given a directed group $G$, a root $r\in V(G)$, an integer $k$ and a set $X\subseteq V$ of terminals,
%the {\em D-Steiner($k,r,X$)} problem is to construct a directed tree rooted at $r$,
%spanning any $k$ terminals in $X$ and of minimum
%possible cost.
%\end{Definition}
%To see the reduction, we replace each undirected edge with two directed edges
%with the same weight and opposite directions.
%For each tree component $T_i$, we add a terminal $t_i$ and edges $(v,t_i)$
%with $c(v,t_i)=H(v|r)\cdot \dist(v,BS)$ for all $v\in T_i$.
%Therefore, by enumerating all possible combinations of the center and $1\leq k\leq n$,
%we can obtain a $n^\epsilon$ approximation
%of the most cost-effective treestar.

%Indeed, it has been shown that the D-Steiner($k,r,X$) problem
%can be approximated  within a factor of $O(n^\epsilon)$ for any
%fixed $\epsilon>0$ \cite{journals/99_directsteiner}.
The min-density group steiner problem can be approximated within a factor of
$O({2\over \epsilon}(\log n)^{2+\epsilon})$
for any constant $\epsilon>0$\cite{journals/06_groupsteiner2}.
The running time is $O(n^{O({1\over\epsilon})})$.
By plugging this result into our greedy framework
%which gives another logarithm factor
and Lemma \ref{lm_wire_restrict}, we obtain the following theorem.

\begin{Theorem}
\label{thm:wnns}There is an algorithm with an approximation factor of
$O({\beta^3\over \epsilon}(\log n)^{3+\epsilon})$
for the WN-NS model for any fixed constant $\epsilon>0$ in
$O(n^{O({1\over \epsilon})})$ time.
\end{Theorem}
}

\choosesecondifkeepwired{}{
\section{The Unicast Model: Poly-Time Algorithm for Restricted Solutions}
\label{sec_unicast}

We present a polynomial time algorithm for computing
the optimal restricted solution under the
unicast communication model,
giving us a $(2+\beta)$-approximation
by Lemma \ref{lm_wire_restrict}. Further we can show that the algorithm will produce an optimal
solution under the uniform entropy and conditional entropy assumption.

\begin{Lemma}
\label{lm_restricted}
For unicast model with uniform entropy and conditional entropy assumption, there is always an optimal solution of the restricted form.
\end{Lemma}
\begin{proof}
We prove the lemma by modifying an optimal compression tree (and the associated
data movement scheme) to a restricted solution without increasing the cost.
Suppose $\calT$ is an optimal compression tree.
We repeatedly process the following types of edges until none is left.
Take an edge $(u,v)\in \calT$ such that $X_v|X_u$ is computed neither on node $u$ nor node $v$ (we call it a {\em bad edge}).
Assume it is computed on node $w$.
Note that we will never change raw data movement flow.
We distinguish two cases:
\begin{enumerate}
\item $w$ is not in the subtree rooted at $v$.
The new compression tree $\calT'$ is formed by deleting $(u,v)$ from $\calT$ and add $(w,v)$ to it.
Instead of $X_v|X_u$,
$X_v|X_w$ is computed on $w$. Everything else is kept unchanged.
It is not hard to see $\calT'$ is a valid compression tree and
the new data movement scheme implements it.

\item $w$ is in the subtree rooted at $v$.
In this case, we delete $(u,v)$ from $\calT$ and add $(u,w)$
to obtain the new compression tree $\calT'$.
Accordingly, all edges in the path from $v$ to $w$ need to change their directions.
The data movement scheme is modified as follows.
Instead of sending $X_v|X_u$
we send $X_w|X_u$ from $w$ ($w$ has $X_u$).
For each edge $(x,y)$ in the path from $v$ to $w$ in $\calT$,
we send $X_y|X_x$ instead of $X_y|X_x$ (at the same location).
This modification corresponds to the change of the direction of $(x,y)$.
It is easy to see these modifications don't change the cost.
\end{enumerate}

It is not hard to see $\calT'$ is also an valid compression tree with one less bad edge in either case.
Therefore, repeating the above process generates a restricted solution with the same cost.
\qed
\end{proof}

%\begin{figure}
%\begin{floatingfigure}[r]{2.6in}
%    \centerline{\includegraphics[width=2.7in]{bidirected}}
%    \caption{(i) Bidirected graph created for the network in Figure \ref{fig:prior-approaches} by the algorithm
%    (not all directed edges are shown). Any
%    directed spanning tree of this graph is a valid compression tree. (ii) The example compression tree from
%    Figure \ref{fig:prior-approaches}(iii).}
%    \vspace{8pt}
%    \label{fig:bidirected}
%\end{figure}
%\end{floatingfigure}

The same problem was previously considered by Cristescu et al.~\cite{cristescu,CBV2004}, who
also propose an approach that uses only second-order distributions (and makes the uniform
entropy and conditional entropy assumption). They develop a $2(2+\sqrt{2})$-approximation
for the problem. However, the solution space we consider in this paper
is larger than the one they consider, in that it allows more freedom in choosing the compression
trees\footnote{Another subtle difference is that, they require all the communication to be along
only one routing tree; we don't require that from our solutions.}.

We note that our approach to find a optimal restricted solution
is essentially the same as the one used by Rickenbach and Wattenhofer \cite{Rickenbach04gatheringcorrelated}.
They also made use of the minimum weight (out-)arborescence algorithm to compute an optimal data collection scheme
under some conditions. Actually, it can be show that their solution space coincides with our {\em restricted solution} space, i.e.,
$X_i|X_j$ should be computed either at $i$ or $j$.

Due to the significant resemblance to \cite{Rickenbach04gatheringcorrelated},
we only briefly sketch our algorithm.
Consider a compression tree $\calT$, and an edge
$(u, v) \in \calT$ where $u$ is the parent of $v$ ($u$ and $v$ may not be adjacent in $\calG_C$).
By induction, we assume that the base station can restore the value of $X_u$ (using its parent).
To compress $X_v$ using the value of $X_u$, we have two options:
%For restoring $X_v$ using the value of $X_u$, we
%and this can be done in either of the following two ways.
\begin{myenumerate}
\item Node $u$ sends the value of $X_u$ ($= x_u$) to $v$, $v$ compresses $X_v$
using the conditional distribution $Pr(X_v | X_u = x_u)$, and sends the
result to the base station.
%and sends the result to the base station.
The cost incurred is $H(X_u) \dist(u,v) + H(X_v|X_u) \dist(v,BS)$.
\item Node $v$ sends $X_v$ to $u$, $u$ compresses $X_v$ given its value of $X_u$, and transmits
the result to the base station $BS$. %and $u$ computes $H(X_v|X_u)$ and sends the result to the base station.
The cost incurred is $H(X_v)\dist(u,v)+H(X_v|X_u)\dist(u,BS)$.
\end{myenumerate}

We observe that the above choice has no impact on restoring information of any
other node and thus it can be made independently for each pair of nodes.

The discussion yields the following algorithm.
Construct a weighted directed graph $G$ with the same vertex set as $\calG_C$.
For each pair of vertices in the communication graph $\calG_C$,
we add {\em two} directed edges. The cost of the directed
edge $(u,v)$ is set to be:
$$ \mbox{\ \ }c(u,v)=\min\{H(X_u)\dist(u,v)+H(X_v|X_u)\dist(v,BS), H(X_v)\dist(u,v)+H(X_v|X_u)\dist(u,BS)\}$$
Similarly we add an edge $(v, u)$.
Essentially, $c(u,v)$ captures the minimal cost incurred in using $X_u$ to
compress $X_v$ (assuming a restricted solution).
Further, we add edges $(BS,v)$ from $BS$ to every node $v$ with cost
$c(BS,v)=H(X_v)\dist(v,BS)$.
%Figure \ref{fig:bidirected} shows an instance
%of this construction for our running example.

We then compute a minimum weight (out-)arborescence $\calT$
(directed spanning tree) rooted at $BS$
which serves as our final compression tree \cite{GGST}. The actual data transmission plan
is easily constructed from the above discussion.
}

\section{Experimental Evaluation}
\label{sec:experiment}
We conducted a comprehensive simulation study over several datasets comparing
the performance of several approaches for data collection. Our results illustrate
that our algorithms can exploit the spatial correlations in the data effectively,
and perform comparably to the DSC lower bound.
\choosesecondifkeepwired{Below we present results over a few representative settings.}
{Due to space constraints, we present results only for the WL model (broadcast communication)
over a few representative settings.}

\topic{Comparison systems:}\\[2pt]
We compare the following data collection methods.
\begin{list}{--}{\leftmargin 0.25in \topsep 0pt \itemsep 0pt}
\item IND (Sec. \ref{sec:prior-approaches}): Each node compresses its data independently
    of the others.
\item Cluster (Sec. \ref{sec:prior-approaches}): The clusters are chosen using the greedy algorithm
        presented in Chu et al.~\cite{CDHH06} -- we start with each node being in its own cluster, and combine clusters
        greedily, till no improvement is observed.
\item DSC: the theoretical lower bound is plotted (Sec. \ref{sec:prior-approaches}).
\item TreeStar: Our algorithm, presented in Sec.~\ref{treestar}, augmented with a greedy local improvement step\footnote{After the TreeStar algorithm
finds a feasible solution, adding a few redundant local broadcasts can cause significant reduction
in the NC cost. We greedily add such local broadcasts till the solution stops improving.}.
%and this causes a large decrease in the conditional entropy being
%sent by some of the adjacent nodes. This improvement can be done as
%long as the change to the compression tree maintains its tree structure.}.
%discussed in Section~\cite{treestar-local}.
\end{list}
\vspace{2pt}

\noindent{For} the TreeStar algorithm, we also show the NC cost (which measures how well the compression tree chosen by TreeStar
        approximates the original distribution).
%entropies. This, in some sense, measures how well the compression tree chosen by TreeStar approximates the original
%distribution.
This cost is lower bounded by the cost of DSC (which uses the best possible compression tree).

\begin{figure*}
    \vspace{-5pt}
    \centerline{\includegraphics[width=6.8in]{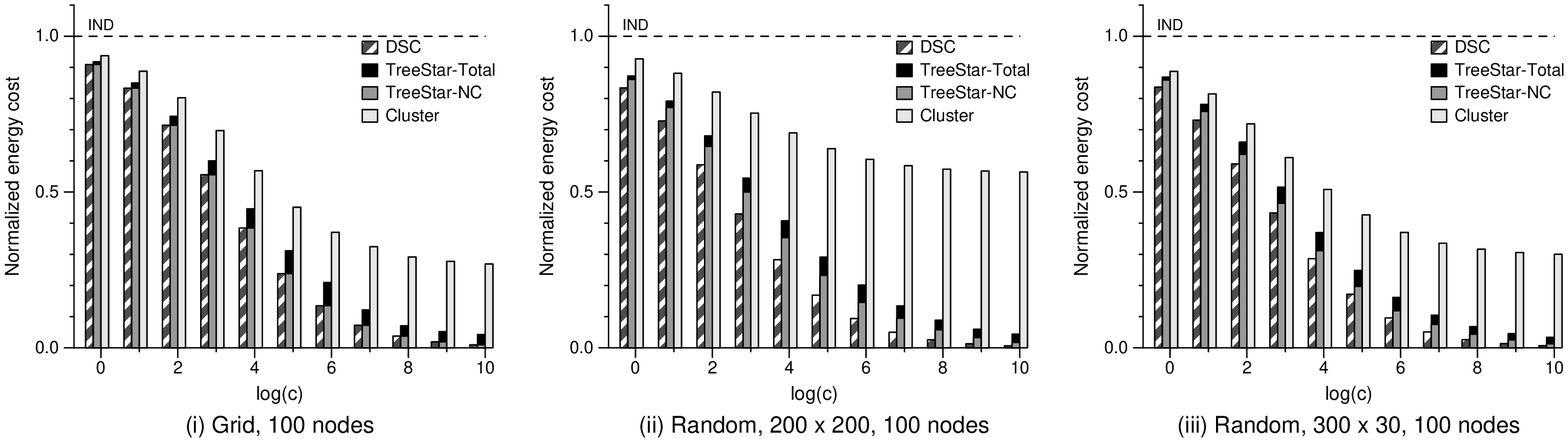}}
    \vspace{-5pt}
    \caption{Results of the experimental evaluation over the Rainfall data}
    \vspace{-5pt}
    \label{fig:rainfall-expt}
\end{figure*}

\topic{Rainfall Data:}\\[2pt]
For our first set of experiments, we use an analytical
expression of the entropy that was derived by Pattem et al.~\cite{PKG04} for a data set containing {\em precipitation} data collected in the
states of Washington and Oregon during 1949-1994~\cite{RainfallData}.
All the nodes have uniform entropy ($H(X_i) = h$), and the conditional entropies are given by: \\[2pt]
\centerline{$H(X_i | X_j) = (1 - \frac{c}{c + dist(i, j)}) h$}
where $dist(i, j)$ is the Euclidean distance between the sensors $i$ and $j$.
The parameter $c$ controls the correlation. For small values of $c$, $H(X_i | X_j) \approx h$ (indicating independence), but
as $c$ increases, the conditional entropy approaches 0.

Figure \ref{fig:rainfall-expt} shows the results for 3 synthetically generated sensor networks. We plot the
total communication cost for each of the above approaches normalized by the cost of IND. The first plot
shows the results for a 100-node network where the sensor nodes are arranged in a uniform grid.
Since the conditional entropies depend only on the distance, for any two adjacent nodes $i$, $j$,
$H(X_i | X_j)$ is constant. Because of this, TreeStar-NC is always equal to DSC in this case.
As we can see, the extra cost (of local broadcasts) is quite small, and overall TreeStar performs
much better than either Cluster or IND, and performs nearly as well as DSC.

We then ran experiments on randomly generated sensor networks, both containing 100 nodes each.
The nodes were randomly placed in either a 200x200 square or a 300x30 rectangle, and communication
links were added between nodes that were sufficiently close to each other ($distance < 30$).
For each plotted data point, we ran the algorithms on 10 randomly chosen networks, and
averaged the results.  As we can see in Figures \ref{fig:rainfall-expt} (ii) and (iii), the relative
performance of the algorithms is quite similar to the first experiment. Note that, because the
conditional entropies are not uniform, TreeStar-NC cost was typically somewhat higher than DSC. The
cost of local broadcasts for TreeStar was again relatively low.

\topic{Gaussian approximation to the Intel Lab Data:}\\[2pt]
For our second set of experiments, we used {\em multivariate Gaussian} models learned over the {\em temperature}
data collected at an indoor, 49-node deployment at the Intel Research Lab, Berkeley\footnote{
\url{http://db.csail.mit.edu/labdata/labdata.html}}.
Separate models were learned for each hour of day~\cite{DGMHH2004} and we show results for 6 of those. After
learning the Gaussian model, we use the {\em differential entropy} of these Gaussians for comparing
the data collection costs. We use the aggregated connectivity data available with the dataset to
simulate different connectivity behavior: in one case, we put communication links between nodes where
the success probability was $> .35$, resulting in somewhat sparse network, whereas in the other
case, we used a threshold of .20.

Figure \ref{fig:gaussian-expt} shows the comparative results for this dataset. The dataset does not exhibit
very strong spatial correlations: as we can see, optimal exploitation of the spatial correlations (using DSC) can only
result in at best a factor of 4 or 5 improvement over IND (which ignores the correlations). However, TreeStar
still performs very well compared to the lower bound on the data collection cost, and much better than
the Cluster approach. Different connectivity behavior does not affect the relative performance of the
algorithms much, with the low-connectivity network consistently incurring about twice as much total energy
cost compared to the high-connectivity network.

\begin{figure}
    \centerline{\includegraphics[width=3.5in]{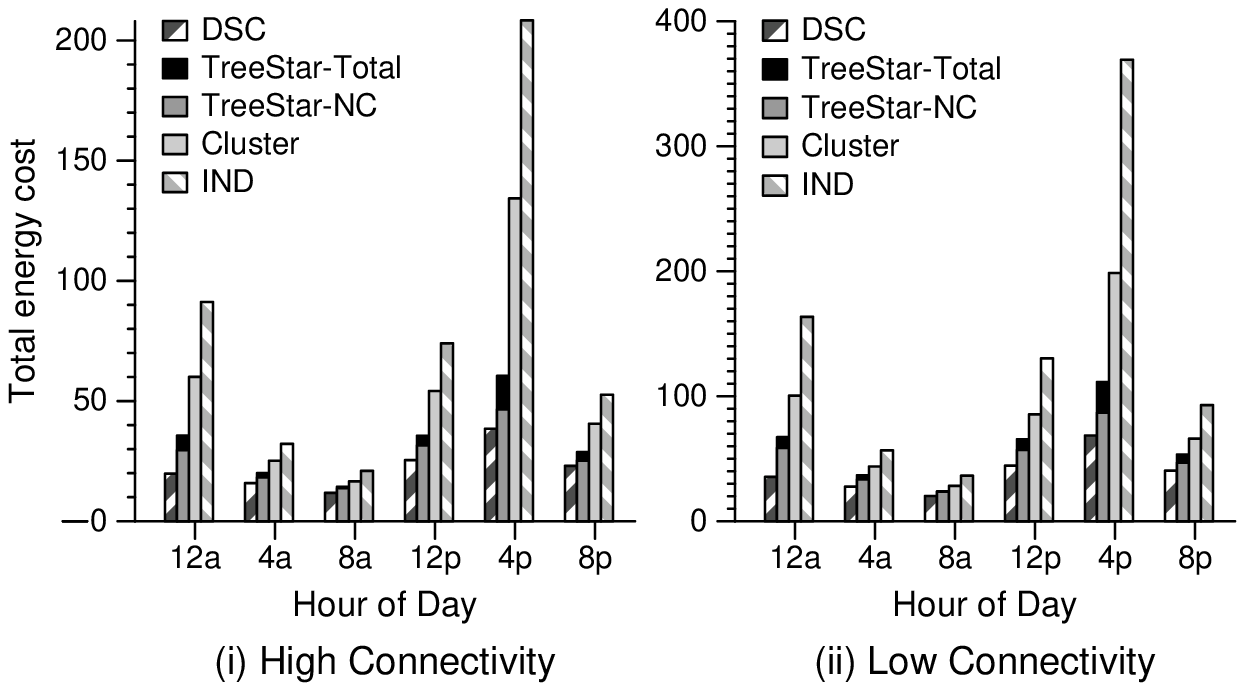}}
    \vspace{-5pt}
    \caption{Results for the Gaussian dataset}
    \label{fig:gaussian-expt}
\end{figure}

\section {Related Work}
\label{sec:relatedwork}
Wireless sensor networks have been a very active area of research in recent
years (see~\cite{akyildiz-survey} for a survey).
Due to space constraints, we only discuss some of the most closely related work
on data collection in sensor networks here.
Directed diffusion~\cite{DirDiffusion}, TinyDB~\cite{tinydb}, LEACH~\cite{leach} are some of
the general purpose data collection mechanisms that have been proposed in the literature. The focus of
that work has been on designing protocols and/or declarative interfaces to
collect data, and not on optimizing continuous data collection.
Aside from the works discussed earlier in the paper~\cite{PKG04,CDHH06,cristescu},
the BBQ system~\cite{DGMHH2004} also uses a
predictive modeling-based approach to collect data from a sensor network.
However, the BBQ system only provides probabilistic, approximate answers to
queries, without any guarantees on the correctness.
Scaglione and Servetto~\cite{SS02} also consider the
interdependence of routing and data compression, but the problem they focus on
(getting all data to all nodes) is different from the problem we address.
%Cristescu et al.~\cite{cristescu} consider the problem of finding a near-optimal
%tree-based communication structure to minimize the total transmission
%cost; their approach is similar to routing driven compression (RDC)~\cite{SS02,PKG04} and may require
%repeated compression and decompression over large numbers of data sources at the sensor nodes,
%which may make it unsuitable for resource-constrained sensor networks.
In seminal work, Gupta and Kumar~\cite{kumar1} proved that the transport
capacity of a random wireless network scales only as $O(\sqrt{n})$, where $n$ is
the number of sensor nodes. Although this seriously limits the scalability of
sensor networks in some domains, in the kinds of applications we are looking at,
the {\em bandwidth} or the {\em rate} is rarely the limiting factor; to be able
to last a long time, the sensor nodes are typically almost always in sleep mode.

Several approaches not based on predictive modeling have also been proposed for
data collection in sensor networks or distributed environments.
For example, constraint chaining~\cite{conch} is a suppression-based exact data collection
approach that monitors a minimal set of node and edge constraints to ensure
correct recovery of the values at the base station.
%Kotidis~\cite{conf/icde/Kotidis05} and Gupta et al.~\cite{conf/mobihoc/GuptaNDC05} consider approaches based on using a
%representative set of sensor nodes to approximate the data distribution over the
%entire network; these approaches however don't solve the problem of exact data
%collection.
%More recently, Cormode et al.~\cite{cormode05} have
%proposed a similar approach of using replicated predictive models to solve the problem of maintaining
%accurate quantile summaries over distributed data sources.

%Much other work, on building routing trees etc. Work by Adler (SODA
%2005) and follow-up work.

\vspace{-4pt}
\section{Conclusions}
\label{sec:conclusion}
Designing practical data collection protocols that can optimally exploit
the strong spatial correlations typically observed in a given sensor network
remains an open problem. In this paper, we considered this problem with the
restriction that the data collection protocol can only utilize second-order
marginal or conditional distributions. We analyzed the problem, and drew
strong connections to the previously studied weakly-connected
dominating set problem. This enabled us to develop a greedy framework for approximating
this problem under various different communication model or solution space settings.
Although we are not able to obtain constant factor approximations,
our empirical study showed that our approach performs very well compared to the
DSC lower bound.
We observe that the worst case for the problem appears to be when the conditional
entropies are close to zero, and that we can get better approximation bounds if
we lower-bound the conditional entropies.
%we assume that the conditional entropies are larger than a constant value.
Future research directions include generalizing our approach to
consider higher-order marginal and conditional distributions, and improving the
approximation bounds by incorporating lower bounds on the conditional entropy
values.
%
%problem of finding an optimal or a near-optimal {\em compression tree} for data collection in sensor
%networks; a compression tree is a directed tree over the nodes of the sensor network such that the value of
%a node is compressed using the value of its parent. We consider this problem under different
%communication models, focusing largely on the {\em broadcast communication} model that enables many
%new opportunities for energy-efficient data collection. We draw connections between this problem and several
%previously studied
%graph problems, and develop novel approximation algorithms for the problem. We present
%comparative results on several synthetic and real-world datasets showing that our algorithms
%construct near-optimal compression trees that yield a significant reduction in the data collection
%cost.
%We
%We considered the problem of optimally exploiting spatio-temporal correlations,
%restricted to using second-order conditional effects. We did a
%comprehensive study under different communication models.
%We developed a polynomial-time algorithm for the case of wired networks with
%point to point communication. For the remaining cases we developed a
%greedy framework that works very nicely across a variety of models,
%and provided approximation algorithms for them.
%Although the approximation bounds are not constant factor approximations,
%as we can see,
%in practice the algorithms perform very well. Our current
%research is focusing on higher order coefficients.
%Eventually we are interested in understanding
%how close we can get to the lower bound given by DSC.

{
\vspace{-4pt}
\bibliographystyle{plain}
\small
\bibliography{all}

\begin{thebibliography}{10}

\bibitem{akyildiz-survey}
I.F. Akyildiz, W.~Su, Y.~Sankarasubramaniam, and E.~Cayirci.
\newblock Wireless sensor networks: a survey.
\newblock {\em Computer Networks}, 2002.

\bibitem{journals/99_directsteiner}
M.~Charikar, C.~Chekuri, T.~Cheung, Z.~Dai, A.~Goel, and M.~Li.
\newblock Approximation algorithm for directed {S}teiner problem.
\newblock {\em Journal of Algorithms}, 33(1):73--91, 1999.

\bibitem{journals/06_groupsteiner2}
C.~Chekuri, G.~Even, and G.~Kortsarz.
\newblock A greedy approximation algorithm for the group {S}teiner problem.
\newblock {\em Discrete Applied Mathematics}, 154(1):15--34, 2006.

\bibitem{Mobihoc2002_chenliestman}
Y.~Chen and A.~L. Liestman.
\newblock Approximating minimum size weakly-connected dominating sets for
  clustering mobile ad hoc networks.
\newblock In {\em Mobihoc}, pages 165--172, 2002.

\bibitem{chowliu}
C.K. Chow and C.N. Liu.
\newblock {Approximating Discrete Probability Distributions with Dependence
  Trees}.
\newblock {\em IEEE Transactions on Information Theory}, (3):462--467, 1968.

\bibitem{CDHH06}
D.~Chu, A.~Deshpande, J.~Hellerstein, and W.~Hong.
\newblock Approximate data collection in sensor networks using probabilistic
  models.
\newblock In {\em International Conference on Data Engineering (ICDE)}, 2006.

\bibitem{CBV2004}
R.~Cristescu, B.~Beferull-Lozano, and M.~Vetterli.
\newblock Networked slepian-wolf: Theory and algorithms.
\newblock In {\em EWSN}, 2004.

\bibitem{cristescu}
R.~Cristescu, B.~Beferull-Lozano, M.~Vetterli, and R.~Wattenhofer.
\newblock Network correlated data gathering with explicit communication:
  Np-completeness and algorithms.
\newblock {\em IEEE/ACM Transactions on Networking}, 14(1):41--54, 2006.

\bibitem{DGMHH2004}
A.~Deshpande, C.~Guestrin, S.~Madden, J.~Hellerstein, and W.~Hong.
\newblock Model-driven data acquisition in sensor networks.
\newblock In {\em VLDB}, 2004.

\bibitem{feige98}
Uriel Feige.
\newblock A threshold of $\ln n$ for approximating set cover.
\newblock {\em J. ACM}, 45(4):634--652, 1998.

\bibitem{GGST}
H.~N. Gabow, Z.~Galil, T.~Spencer, and R.~E. Tarjan.
\newblock Efficient algorithms for finding minimum spanning trees in undirected
  and directed graphs.
\newblock {\em Combinatorica}, 6(2):109--122, 1986.

\bibitem{GE03}
A.~Goel and D.~Estrin.
\newblock Simultaneous optimization for concave costs: Single sink aggregation
  or single source buy-at-bulk.
\newblock In {\em {ACM}-{SIAM} Symposium on Discrete Algorithms ({SODA})},
  2003.

\bibitem{journals/algorithmica/GuhaK98}
S.~Guha and S.~Khuller.
\newblock Approximation algorithms for connected dominating sets.
\newblock {\em Algorithmica}, 20(4), 1998.

\bibitem{conf/mobihoc/GuptaNDC05}
H.~Gupta, V.~Navda, S.~Das, and V.~Chowdhary.
\newblock Efficient gathering of correlated data in sensor networks.
\newblock In {\em MobiHoc}, 2005.

\bibitem{kumar1}
P.~Gupta and P.~R. Kumar.
\newblock The capacity of wireless networks.
\newblock {\em IEEE Transactions on Information Theory}, 46, 2000.

\bibitem{leach}
W.~R. Heinzelman, A.~Chandrakasan, and H.~Balakrishnan.
\newblock Energy-efficient communication protocol for wireless microsensor
  networks.
\newblock In {\em HICSS}, 2000.

\bibitem{DirDiffusion}
C.~Intanagonwiwat, R.~Govindan, and D.~Estrin.
\newblock Directed diffusion: A scalable and robust communication paradigm for
  sensor networks.
\newblock In {\em {ACM} {MobiCOM}}, 2000.

\bibitem{conf/icde/Kotidis05}
Y.~Kotidis.
\newblock Snapshot queries: Towards data-centric sensor networks.
\newblock In {\em ICDE}, 2005.

\bibitem{conf/mobicom/Liu06}
J.~Liu, M.~Adler, D.~Towsley, and C.~Zhang.
\newblock On optimal communication cost for gathering correlated data through
  wireless sensor networks.
\newblock In {\em Proceedings of {ACM} {MobiCOM}}, 2006.

\bibitem{tinydb}
Samuel Madden, Wei Hong, Joseph~M. Hellerstein, and Michael Franklin.
\newblock {TinyDB} web page.
\newblock http://telegraph.cs.berkeley.edu/tinydb.

\bibitem{PKG04}
S.~Pattem, B.~Krishnamachari, and R.~Govindan.
\newblock The impact of spatial correlation on routing with compression in
  wireless sensor networks.
\newblock In {\em IPSN}, 2004.

\bibitem{pradhan}
S.~Pradhan and K.~Ramchandran.
\newblock Distributed source coding using syndromes ({DISCUS}): Design and
  construction.
\newblock {\em IEEE Trans.\ Information Theory}, 2003.

\bibitem{SS02}
A.~Scaglione and S.~Servetto.
\newblock On the interdependence of routing and data compression in multi-hop
  sensor networks.
\newblock In {\em Mobicom}, 2002.

\bibitem{conch}
A.~Silberstein, R.~Braynard, and J.~Yang.
\newblock Constraint-chaining: On energy-efficient continuous monitoring in
  sensor networks.
\newblock In {\em SIGMOD}, 2006.

\bibitem{SW1973}
D.~Slepian and J~Wolf.
\newblock Noiseless coding of correlated information sources.
\newblock {\em IEEE Transactions on Information Theory}, 19(4), 1973.

\bibitem{journals/tosn/Su07}
Xun Su.
\newblock A combinatorial algorithmic approach to energy efficient information
  collection in wireless sensor networks.
\newblock {\em ACM Trans. Sen. Netw.}, 3(1):6, 2007.

\bibitem{Rickenbach04gatheringcorrelated}
Pascal von Rickenbach and Roger Wattenhofer.
\newblock Gathering correlated data in sensor networks.
\newblock In {\em In Proc. of the ACM Joint Workshop on Foundations of Mobile
  Computing (DIALM-POMC)}, pages 60--66, 2004.

\bibitem{wang}
L.~Wang and A.~Deshpande.
\newblock Predictive modeling-based data collection in wireless sensor
  networks.
\newblock In {\em EWSN}, 2008.

\bibitem{RainfallData}
M.~Widmann and C.~Bretherton.
\newblock 50 km resolution daily precipitation for the pacific northwest, 2003.
\newblock http://www.jisao.washington.edu/data\_sets/widmann.

\bibitem{wyner}
A.~D. Wyner and J.~Ziv.
\newblock The rate-distortion function for source coding with side information
  at the decoder.
\newblock {\em IEEE Transactions on Information Theory}, 1976.

\bibitem{xiong}
Z.~Xiong, A.~D. Liveris, and S.~Cheng.
\newblock Distributed source coding for sensor networks.
\newblock {\em IEEE Signal Processing Magazine}, 21, 2004.

\end{thebibliography}
}
\end{document}